\documentclass[showpacs,amsmath,amssymb,twocolumn,superscriptaddress,notitlepage,preprintnumbers,pra]{revtex4-1}
\usepackage[dvips]{graphicx}
\usepackage{subfigure}
\usepackage{amsmath,amssymb,amsthm,mathrsfs,amsfonts,dsfont}
\usepackage{dcolumn}
\usepackage{bm}
\usepackage{braket}
\usepackage{physics}
\usepackage{natbib}
\usepackage{mathtools}
\usepackage{diagbox}
\usepackage[utf8]{inputenc}
\usepackage[english]{babel}
\usepackage{scalerel}[2014/03/10]
\usepackage{stackengine}
\usepackage[noend,noline,ruled,linesnumbered]{algorithm2e}
\usepackage[noend]{algpseudocode} 
\usepackage{makecell}
\usepackage{footnotebackref}
\usepackage{lipsum}
\usepackage{hyperref}
\usepackage{multirow}
\usepackage{makecell}
\hypersetup{colorlinks=true, linkcolor=blue, citecolor=blue, urlcolor=black }
\makeatletter
\newtheorem{assumption}{Assumption}
\newtheorem{lemma}{Lemma}
\newtheorem{theorem}{Theorem}
\newtheorem{definition}{Definition}
\newtheorem{remark}{Remark}

\begin{document}
\title{Trainability Analysis of Quantum Optimization Algorithms from a Bayesian Lens}
\author{Yanqi Song}
\affiliation{State Key Laboratory of Networking and Switching Technology, Beijing University of Posts and Telecommunications, Beijing, 100876, China
}
\author{Yusen Wu}
\affiliation{
Department of Physics, The University of Western Australia, Perth, WA 6009, Australia
}
\author{Sujuan Qin}
\affiliation{State Key Laboratory of Networking and Switching Technology, Beijing University of Posts and Telecommunications, Beijing, 100876, China
}
\author{Qiaoyan Wen}
\affiliation{State Key Laboratory of Networking and Switching Technology, Beijing University of Posts and Telecommunications, Beijing, 100876, China
}
\author{Jingbo B. Wang}
\email{jingbo.wang@uwa.edu.au}
\affiliation{
Department of Physics, The University of Western Australia, Perth, WA 6009, Australia
}
\author{Fei Gao}
\email{gaof@bupt.edu.cn}
\affiliation{State Key Laboratory of Networking and Switching Technology, Beijing University of Posts and Telecommunications, Beijing, 100876, China
 }

\date{\today}

\begin{abstract}
The Quantum Approximate Optimization Algorithm (QAOA) is an extensively studied variational quantum algorithm utilized for solving optimization problems on near-term quantum devices. A significant focus is placed on determining the effectiveness of training the $n$-qubit QAOA circuit, i.e., whether the optimization error can converge to a constant level as the number of optimization iterations scales polynomially with the number of qubits. In realistic scenarios, the landscape of the corresponding QAOA objective function is generally non-convex and contains numerous local optima. 
In this work, motivated by the favorable performance of Bayesian
optimization in handling non-convex functions, we theoretically investigate the trainability of the QAOA circuit through the lens of the Bayesian approach. This lens considers the corresponding QAOA objective function as a sample drawn from a specific Gaussian process. Specifically, we focus on two scenarios: the noiseless QAOA circuit and the noisy QAOA circuit subjected to local Pauli channels. Our first result demonstrates that the noiseless QAOA circuit with a depth of $\tilde{\mathcal{O}}\left(\sqrt{\log n}\right)$ can be trained efficiently, based on the widely accepted assumption that either
the left or right slice of each block in the circuit forms a local 1-design. Furthermore, we show that if each quantum gate is affected by a $q$-strength local Pauli channel with the noise strength range of $1/{\rm poly} (n)$ to 0.1, the noisy QAOA circuit with a depth of $\mathcal{O}\left(\log n/\log(1/q)\right)$ can also be trained efficiently. Our results offer valuable insights into the theoretical performance of quantum optimization algorithms in the noisy intermediate-scale quantum era.
\end{abstract}

\maketitle

\section{Introduction}
The Variational Quantum Algorithm (VQA)~\cite{mcclean2016theory,moll2018quantum,cerezo2021variational,bharti2022noisy} is expected to exhibit potential quantum advantages in the Noisy Intermediate-Scale Quantum (NISQ) era~\cite{preskill2018quantum}. It has been successfully applied to various areas, such as quantum chemistry~\cite{peruzzo2014variational,o2016scalable,kandala2017hardware,romero2018strategies,hempel2018quantum,grimsley2019adaptive,google2020hartree,yuan2021quantum,wu2023orbital,fujii2022deep}, machine learning~\cite{farhi2018classification,mitarai2018quantum,schuld2019quantum,havlivcek2019supervised,schuld2020circuit, wu2023quantum}, and constraint satisfaction problems~\cite{farhi2014quantum,wang2018quantum,hadfield2019quantum,zhou2020quantum,pagano2020quantum,harrigan2021quantum,ebadi2022quantum,qiang2016efficient,qiang2018large,marsh2020combinatorial,qiang2021implementing}. One notable instance of VQA is the Quantum Approximate Optimization Algorithm (QAOA)~\cite{farhi2014quantum}, designed to solve constraint satisfaction problems involving the optimization of a quadratic function of binary variables~\cite{lucas2014ising}. QAOA encodes the quadratic function onto a problem-oriented Hamiltonian and searches for its optimal solution using a hybrid quantum-classical algorithm. The quantum subroutine of QAOA generates parameterized quantum states using the QAOA circuit and estimates the QAOA objective function through multi-round measurements, while the classical optimization method iteratively updates the variational parameters within the QAOA circuit. 
 
Understanding the time complexity of QAOA is crucial for assessing its potential quantum advantages. The time complexity of QAOA depends on the QAOA circuit depth, measurement complexity, and optimization iteration complexity. Extensive research has been conducted on the theory of required QAOA circuit depth~\cite{bravyi2020obstacles,farhi2020quantum,herrman2021lower,stilck2021limitations,you2022convergence,binkowski2023elementary} and measurement complexity~\cite{izmaylov2019unitary,zhao2020measurement,huang2020predicting,crawford2021efficient,harrow2021low,wu2022estimating,hadfield2022measurements}. 
However, the potential of achieving quantum advantages through QAOA remains uncertain due to the absence of rigorous analysis on optimization iteration complexity. Previous studies have investigated and implemented various classical optimization methods in QAOA, such as gradient-based methods~\cite{harrow2021low,guerreschi2017practical,sweke2020stochastic,stokes2020quantum,koczor2022quantum} and gradient-free methods~\cite{zhu2019training,self2021variational,tibaldi2022bayesian}. Meanwhile, the theoretical convergence performance of Stochastic Gradient Descent (SGD)~\cite{kingma2014adam} is studied in some specific settings, where the QAOA objective function is strongly-convex~\cite{harrow2021low} or satisfies Polyak-Lojasiewicz inequality~\cite{sweke2020stochastic}. However, in realistic scenarios, the landscape of the QAOA objective function is generally non-convex with many local optima~\cite{shaydulin2019multistart, huembeli2021characterizing}. Additionally, this landscape may be susceptible to the barren plateau phenomenon when the QAOA circuit reaches a specific depth~\cite{mcclean2018barren, cerezo2021cost, wang2021noise}. Training a deep QAOA circuit in realistic scenarios poses significant challenges. Therefore, it is crucial to address the open question: \emph{what is the depth range of the QAOA circuit that can be effectively trained?} Investigating this question has important implications for analyzing the optimization iteration complexity.

In this work, we investigate this question within the context of Bayesian Optimization (BO)~\cite{snoek2012practical, shahriari2015taking, frazier2018tutorial}, which is designed for gradient-free global optimization. BO is particularly suitable in situations where estimating the objective function is costly and demonstrates its favorable performance in handling the non-convex objective function. Motivated by the excellent numerical performance of BO in QAOA~\cite{zhu2019training, self2021variational, tibaldi2022bayesian}, we theoretically investigate the trainability of the QAOA circuit through the lens of the Bayesian approach. This lens considers the QAOA objective function as a sample drawn from a specific Gaussian process, without relying on explicit assumptions about the landscape of the QAOA objective function. It offers valuable insights into the QAOA objective function and the corresponding theoretical analysis may bridge the gap between theory and practice. In detail, recent studies~\cite{bouland2019complexity, wu2023complexity, fontana2023classical} have demonstrated that the QAOA objective function exhibits the properties of a high-order differentiable function in both noiseless and noisy scenarios. Intriguingly, a Gaussian process with the Matern covariance function has the remarkable ability to generate high-order differentiable functions~\cite{rasmussen2006gaussian, kanagawa2018gaussian}. Thus, it is plausible to interpret the QAOA objective function as a sample drawn from a Gaussian process with the Matern covariance function, regardless of the presence of noise. Exploiting the profound links between BO and information theory~\cite{srinivas2012information,janz2020bandit,vakili2021information}, we investigate the trainability of the $n$-qubit QAOA circuit through the aforementioned lens in both noiseless and noisy scenarios. We start by analyzing the continuity property of the QAOA objective function. By leveraging this property, we establish a theoretical limit on the QAOA circuit depth, guaranteeing the convergence of the optimization error to a constant level as the number of optimization iterations scales polynomially with the number of qubits. Specifically, we have obtained the following results. Based on the widely accepted assumption that either the left or right slice of each block in the noiseless QAOA circuit forms a local 1-design~\cite{cerezo2021cost,harrow2009random}, we demonstrate that the circuit with a depth of $\tilde{\mathcal{O}}(\sqrt{\log n})$ can be trained efficiently. Furthermore, in a practical scenario concerning the Maximum Cut problem on an unweighted regular graph, we make a commonly employed assumption that each quantum gate is affected by a $q$-strength local Pauli channel, and the effects of these noises are postponed until the end of each block in the QAOA circuit~\cite{wang2021noise,quek2022exponentially}. We show that if the noise strength range spans $1/{\rm poly} (n)$ to $1/n^{1/\sqrt{\log n}}$, the noisy QAOA circuit with a depth of $\mathcal{O}\left(\log n/\log(1/q)\right)$ can also be trained efficiently. For a more intuitive description of the noise strength range described above, we focus on near-term quantum devices with 50-100 qubits~\cite{preskill2018quantum}. In this case, $1/n^{1/\sqrt{\log n}}$ is only slightly larger than 0.1. This suggests that this range corresponds to the actual noise levels in near-term quantum devices and holds practical significance. Notably, $\mathcal{O}\left(\log n/\log(1/q)\right)$ grows as $q$ increases, implying the growth in the depth of the noisy QAOA circuit that can be effectively trained as the noise strength increases. This finding implies that a specific strength noise may obscure certain local optima, thereby improving the trainability of the noisy QAOA circuit. In conclusion, we consider these two results to be a substantial advancement in understanding the complexity of the iterations needed to optimize the variational parameters of the QAOA circuit, particularly within the NISQ era.

This work is organized as follows: Section~\ref{sec:QAOA} provides a comprehensive review of QAOA. Section~\ref{sec:BO for QAOA} presents a detailed description of BO in the context of QAOA. In Section~\ref{sec:Theoretical Analysis 1}, we offer a theoretical analysis of the trainability of the noiseless QAOA circuit through the lens of the Bayesian approach. Lastly, Section~\ref{sec:Theoretical Analysis 2} theoretically analyzes the trainability of the noisy QAOA circuit with local Pauli channels in a practical scenario.

\section{Quantum Approximate Optimization Algorithm}
\label{sec:QAOA}
In theoretical computational science, constraint satisfaction problems encompass a wide range of typical problems, such as Maximum Cut, Maximum Independent Set, and Graph Coloring~\cite{gross2018graph}. These problems define their constraints as clauses, with a candidate solution represented by a specific assignment of the corresponding binary variables. The objective of these problems is to find an optimal assignment that maximizes the number of satisfied clauses. In other words, solving a constraint satisfaction problem can be reformulated as optimizing a quadratic function involving binary variables. However, finding the exact solution is widely recognized as an 
$\rm NP$-hard problem~\cite{garey1979computers}. Consequently, an alternative approach is to seek an approximate solution.  Inspired by the quantum annealing process~\cite{kadowaki1998quantum}, QAOA was proposed and applied to solve constraint satisfaction problems. Although the prospects of achieving quantum advantages through QAOA remain unclear, it provides a simple paradigm for optimization that can be implemented on near-term quantum devices.

Given a specific constraint satisfaction problem with \(n\) binary variables and $\mathcal{C}$ clauses, QAOA starts by constructing a problem-oriented Hamiltonian \(H_1\), a mixing Hamiltonian \(H_2\), and $2p$ variational parameters $\bm{\theta}=[\theta _{1,1}, \theta_{1,2}, \cdots, \theta_{p,1}, \theta_{p,2}]^{\mathsf{T}}$. Specifically, the problem-oriented Hamiltonian \(H_1\) is a linear combination of $\mathcal{C}$ Pauli strings
\begin{align}\label{eq:H1}
H_1 =\sum_{c=1}^\mathcal{C}\gamma_c P^{n}_c,
\end{align}
where $\gamma_c\in\mathbb{R}$ and $P^{n}_c\in \{\mathbb{I},\sigma^{z}\}^{\otimes n}$ with $\sigma ^z$ is the Pauli \(Z\) operator. The typical form of the mixing Hamiltonian \(H_2\) is the transverse field Hamiltonian 
\begin{align}\label{eq:H2}
H_2 =\sum_{i=1}^{n}\sigma^x_i,
\end{align}
where $\sigma ^x_i$ is the Pauli $X$ operator acting on the $i$-th qubit. Subsequently, by iteratively applying $H_1$ and $H_2$ to the initial state $\rho$ for $p$ rounds, the noiseless QAOA objective function $f(\bm{\theta})$ in the absence of quantum gate noise is given by the following expectation value
\begin{align}\label{eq:Objective Function 1}
f(\bm\theta)=\mathrm {Tr}\left[H_1U(\bm \theta)\rho U^{\dagger}(\bm \theta)\right],
\end{align}
where $\rho=(|+\rangle\langle+|)^{\otimes n}$ denotes the uniform superposition over computational basis states and the noiseless QAOA circuit
\begin{align}\label{eq:Quantum Circuit}
U(\bm\theta)=\prod_{j=1}^p \prod_{l=1}^2 U_{j,l}(\theta_{j,l})
\end{align}
with $U_{j,l}(\theta_{j,l})=\exp(-i\theta_{j,l} H_l)$ for $(j,l)\in[p]\times[2]$. The statistical estimation of \(f(\bm{\theta})\) can be achieved by repeating the aforementioned process with identical parameters and computational basis measurements. After defining \(f(\bm{\theta})\), the next step involves iteratively updating $\bm\theta$ within $U(\bm\theta)$ through classical optimization methods to maximize \(f(\bm{\theta})\) and obtain the global maximum point
\begin{align} \label{eq:Max}
\bm\theta^{*}=\arg\max_{\bm \theta \in\mathcal{D}}f(\bm\theta),
\end{align}
where the domain $\mathcal{D}=[0,2\pi]^{2p}$. 

The classical optimization method is crucial in QAOA, as previously mentioned. Finding \(\bm\theta^{*}\) may be intractable due to the non-convex landscape of \(f(\bm{\theta})\) and the presence of numerous local optima. Thus, it is pivotal to determine the appropriate classical optimization method that can efficiently find a better approximation
of \(\bm\theta^{*}\). In the following section, we illustrate the utilization of the Bayesian approach to accomplish this optimization task and provide a theoretical analysis of the trainability of QAOA.

\section{Optimizing QAOA through the Bayesian Approach}
\label{sec:BO for QAOA}
BO is designed for gradient-free global optimization. It is particularly suitable in situations where estimating the objective function is computationally expensive and the convexity property of the objective function is uncertain~\cite{snoek2012practical,shahriari2015taking,frazier2018tutorial}. BO comprises two essential components: (\romannumeral1)~a statistical model, usually the Gaussian process~\cite{rasmussen2006gaussian}, that generates a posterior distribution conditioned on a prior distribution and a collection of observations of the objective function. (\romannumeral2)~an acquisition function that utilizes the current posterior distribution for the objective function to determine the location of the next point. In the context of QAOA, we present a comprehensive introduction to BO, focusing on the Gaussian process and the acquisition function. Further details regarding each step of BO can be found in Algorithm~\ref{alg:BO for QAOA}.

\begin{algorithm*}[t]
\caption{\textbf{Optimizing QAOA through the Bayesian approach}}\label{alg:BO for QAOA}
\DontPrintSemicolon
\SetKwInOut{Input}{Input}
\SetKwInOut{Output}{Output}
\SetKwInOut{Initialize}{Initialize}
\SetKwFor{While}{while}{do}{}
\SetKwIF{If}{ElseIf}{Else}{if}{then}{else if}{else}{end if}
\Input{the $n$-qubit initial state $\rho=(|+\rangle\langle+|)^{\otimes n}$, the noiseless QAOA circuit $U(\bm\theta)$~(Eq.~\ref{eq:Quantum Circuit}), the problem-oriented Hamiltonian $H_1$~(Eq.~\ref{eq:H1}), the fixed number of measurements $M$, the prior mean function $\mu(\bm\theta)=0$, the prior covariance function $k(\bm\theta, \bm\theta^{\prime})=k_{\rm Matern-\nu}(\bm\theta,\bm\theta^{\prime})$~(Eq.~\ref{eq:Matern}), the domain $\mathcal{D}=[0,2\pi]^{2p}$, the number of initial observations $T_0$, and the number of steps $T$}
\Output {the approximation
of the maximum point $\bm\theta_{T}^+$}
\Initialize {the accumulated observations $\mathcal{S}_0$ with $T_0$ randomly selected points, the approximation
of the maximum point \(\bm\theta_{0}^+\gets\arg\max_{\bm\theta\in\mathcal{S}_{0}}y(\bm\theta)\), the posterior mean function $\mu_{0}(\bm\theta)$ and the posterior variance $\sigma^2_{0}(\bm\theta)$ based on \(\mathcal{S}_{0}\) as described in Eq.~\ref{eq:Posterior}}
\While{$1 \leq t\leq T$}{
\qquad
Calculate the acquisition function $\mathrm{UCB}_{t}(\bm\theta)$ as described in Eq.~\ref{eq:UCB} \;
\qquad Select the next point $\bm\theta_{t}\gets\arg \max_{\bm\theta\in\mathcal{D}}\mathrm{UCB}_{t}(\bm\theta)$ \;
\qquad Calculate $y(\bm\theta_{t})$ by estimating $\mathrm {Tr}\left[H_1U(\bm \theta_t)\rho U^{\dagger}(\bm \theta_t)\right]$ with $M$ measurements\;
\qquad Set $\mathcal{S}_{t}\gets\mathcal{S}_{t-1}\cup \left(\{(\bm\theta_t, y(\bm\theta_{t}) \right\})$\;
 \qquad Update $\mu_t(\bm\theta)$ and $\sigma^2_t(\bm\theta)$ based on \(\mathcal{S}_{t}\) as described in Eq.~\ref{eq:Posterior}\; 
 \qquad\If{$y(\bm\theta_{t})>y(\bm\theta_t^+)$}{
 \qquad\qquad Set \(\bm\theta_t^+\gets\bm\theta_{t}\)}
\qquad Set $t\gets t+1$}
\Return $\bm\theta_{T}^+$
\end{algorithm*}

\subsection{Gaussian Process}
A Gaussian process is a collection of random variables, where any subset forms a multivariate Gaussian distribution. In the optimization task described by Eq.~\ref{eq:Max}, the random variables correspond to the values of the noiseless QAOA objective function \(f(\bm\theta)\) at points \(\bm\theta\in\mathcal{D}\). A Gaussian process, serving as a distribution for \(f(\bm\theta)\), 
is fully determined by its mean function $\mu(\bm\theta)$ and covariance function $k(\bm\theta, \bm\theta^{\prime})$. Specifically, $\mu(\bm\theta)$ specifies the mean value of $f(\bm\theta)$ at any point $\bm\theta$, while $k(\bm\theta, \bm\theta^{\prime})$ determines the covariance between $f(\bm\theta)$ and $f(\bm\theta^{\prime})$ at any two points $\bm\theta$ and $\bm\theta^{\prime}$. The Gaussian process is denoted as 
\begin{align}
f(\bm\theta)\sim\mathcal{GP}(\mu(\bm\theta) , k(\bm\theta,\bm\theta^{\prime})),
\end{align}
where $\mu(\bm\theta)=\mathbb{E}[f(\bm\theta)]$ and $k(\bm\theta, \bm\theta^{\prime})=\mathbb{E}[(f(\bm\theta)-\mu(\bm\theta))(f(\bm\theta^{\prime})-\mu(\bm\theta^{\prime}))]$. It is commonly assumed that the prior mean function \(\mu(\bm\theta)=0\). Additionally, the prior covariance function $k(\bm\theta,\bm\theta^{\prime})$ is commonly chosen from some notable covariance functions, such as the Matern covariance function $k_{\rm Matern-\nu}(\bm\theta,\bm\theta^{\prime})$, whose specific form is provided in Appendix~\ref{appendix:Matern}. 

Suppose we have the following accumulated observations after \(t-1\) steps of BO
\begin{equation}\label{eq:S}
    \mathcal{S}_{t-1}=\{(\bm\theta_{1},y(\bm\theta_{1})), \cdots, (\bm\theta_{t-1}, y(\bm\theta_{t-1}))\},
\end{equation}
where $y(\bm\theta_{i})$ denotes the estimation of $f(\bm\theta_i)$ for $i\in[t-1]$. In each step $i$, the measurements are taken to ensure that $y(\bm\theta_i)= f(\bm\theta_i) + \xi_i^{\rm noise}$, where $\xi_i^{\rm {noise}}\sim N(0,1/4M)$ is independent and identically distributed Gaussian noise with $M$ representing the fixed number of measurements~\footnote{Let $M$ be the fixed number of measurements and $y_j(\bm\theta)$ be the one-shot measurement result for $j\in [M]$. According to Central Limit Theorem~\cite{fischer2011history}, for a sufficiently large $M$, we have $\frac{1}{M}\sum_{j=1}^My_j(\bm\theta)=\mu+\frac{\sigma}{\sqrt{M}}Y$, where $Y\sim N(0,1)$. Here, $\mu$ and $\sigma$ represent the mean value and standard deviation of $y_j(\bm\theta)$ for $j\in[M]$. Consequently, we find $\xi^{\rm noise}=\frac{1}{M}\sum_{j=1}^My_j(\bm\theta)-\mu=\frac{\sigma}{\sqrt{M}}Y$. Thus, the result $\xi^{\rm noise}\sim N(0,\sigma^2/M)$ holds. Considering $\sigma^2\leq1/4$, it is reasonable to assume that $\xi^{\rm noise}\sim N(0,1/4M)$.}.
Conditioned on \(\mathcal{S}_{t-1}\), the distribution for \(f(\bm\theta)\) is a Gaussian process with the posterior mean function \(\mu_{t-1}(\bm\theta)=\mathbb{E}[f(\bm\theta)|\mathcal{S}_{t-1}]\) and the posterior covariance function \(k_{t-1}(\bm\theta,\bm\theta^{\prime})=\mathbb{E}[(f(\bm\theta)-\mu(\bm\theta))(f(\bm\theta^{\prime})-\mu(\bm\theta^{\prime}))|\mathcal{S}_{t-1}]\). These are specified as follows
\begin{equation}\label{eq:Posterior}
    \begin{split}
       \mu_{t-1}(\bm\theta)&=\bm k_{t-1}(\bm\theta)^{\mathsf{T}}(\bm K_{t-1}+\bm I_{t-1}/4M)^{-1}\bm{y} _{t-1},\\
k_{t-1}(\bm\theta,\bm\theta^{\prime})&=k_{\rm Matern-\nu}(\bm\theta,\bm\theta^{\prime})-\\&\bm k_{t-1}(\bm\theta)^{\mathsf{T}}(\bm K_{t-1}+\bm I_{t-1}/4M)^{-1}\bm k_{t-1}(\bm\theta^{\prime}), 
    \end{split}
\end{equation}
where \(\bm k_{t-1}(\bm\theta)=[k_{\rm Matern-\nu}(\bm\theta,\bm\theta_1)  \cdots k_{\rm Matern-\nu}(\bm\theta,\bm\theta_{t-1})]^{\mathsf{T}}\), the positive definite covariance matrix \(\bm K_{t-1}=[k_{\rm Matern-\nu}(\bm\theta_i, \bm\theta_j)]_{\bm\theta_i,\bm\theta_j\in \mathcal{A}_{t-1}}\) with the accumulated points \(\mathcal{A}_{t-1}=\{\bm\theta_1,\cdots,\bm\theta_{t-1}\}\) from the previous $t-1$ steps, and \(\bm{ y}_{t-1}=[y(\bm\theta_1) \cdots y(\bm\theta_{t-1})]^{\mathsf{T}}\). The posterior variance of \(f(\bm\theta)\) is denoted as \(\sigma^2_{t-1}(\bm\theta)=k_{t-1}(\bm\theta,\bm\theta)\). 

\begin{figure*}[htpb]
\centering
\begin{minipage}{16.8cm} 
\includegraphics[width=0.99\textwidth]
{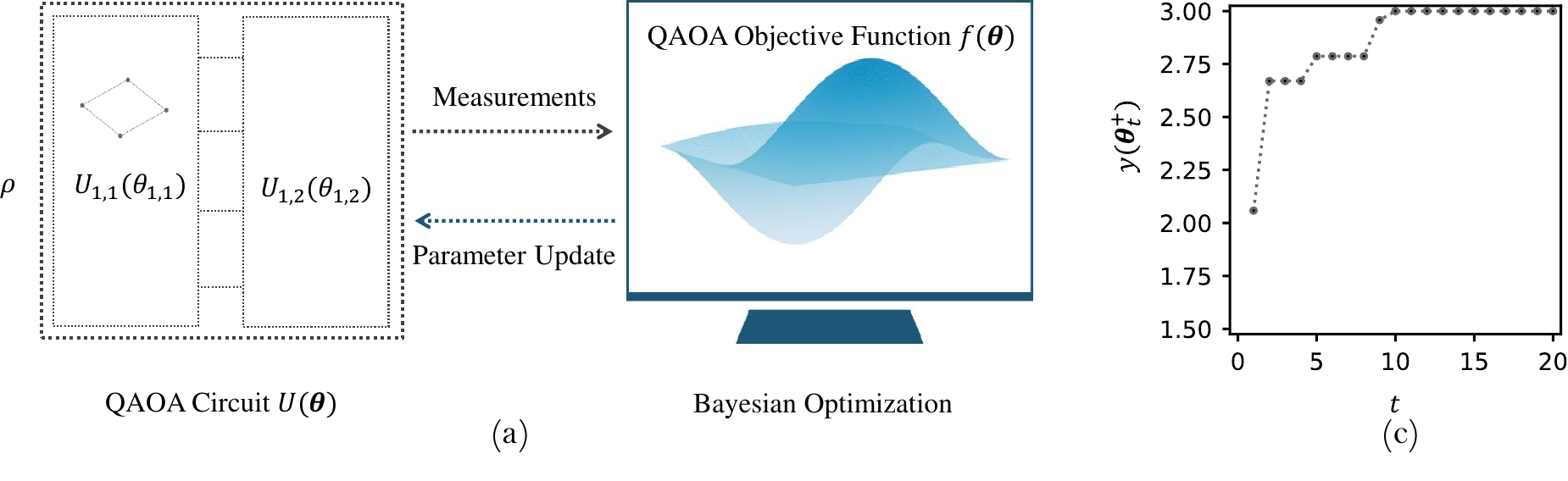} 
\end{minipage}
\begin{minipage}{16.8cm}
\includegraphics[width=0.99\textwidth]{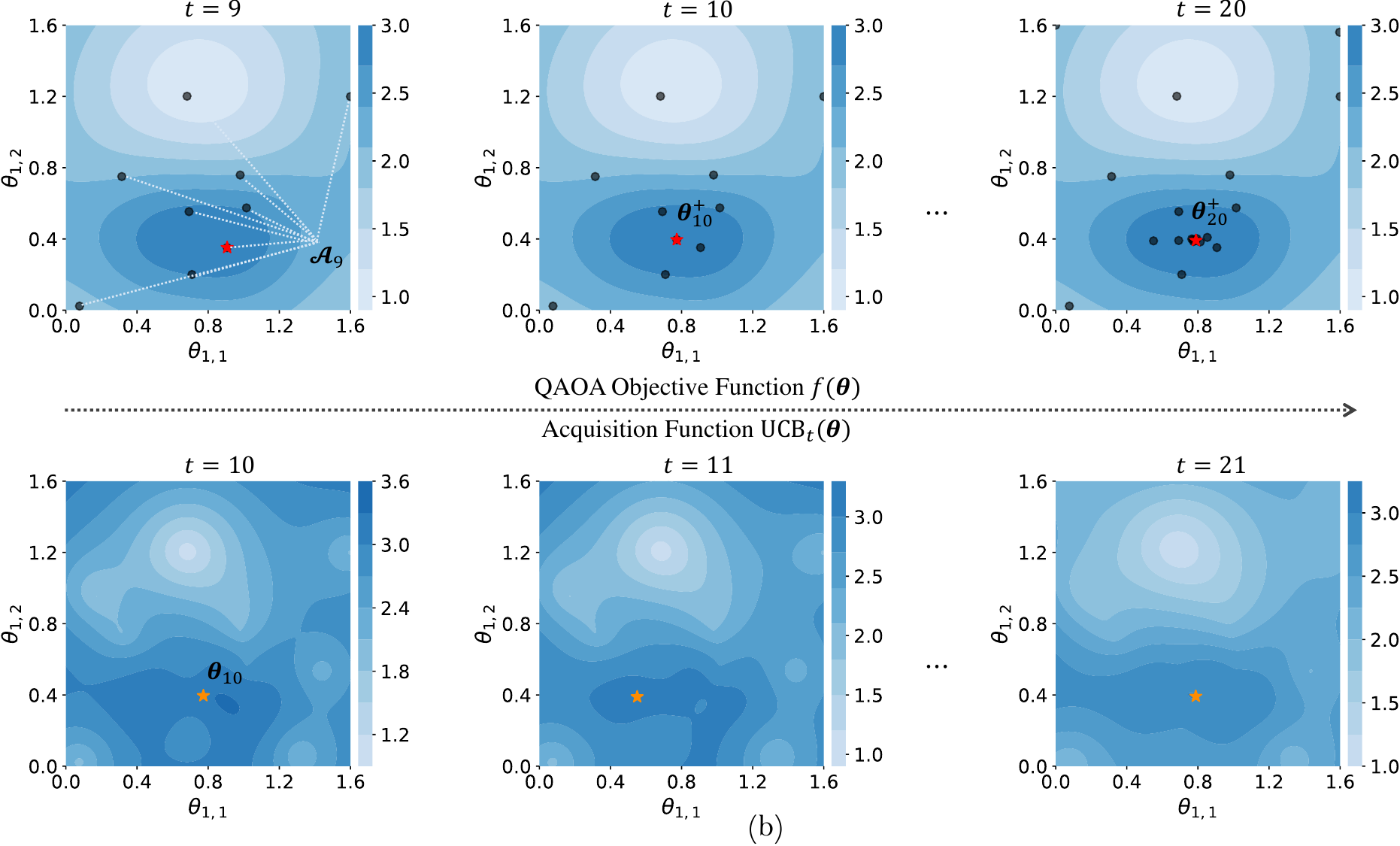} 
\end{minipage}
\caption{\textbf{Numerical performance of QAOA for the Maximum Cut problem using the Bayesian approach.} 
\textbf{(a)}~Structure of QAOA on a Maximum Cut graph instance. For a 2-regular graph with 4 vertices, the quantum subroutine prepares parameterized quantum states using a 1-layer noiseless QAOA circuit $U(\bm\theta)$. It estimates the corresponding two-dimensional noiseless QAOA objective function $f(\bm\theta)$ with $\bm\theta=(\theta_{1,1},\theta_{1,2})$ through multi-round measurements. BO iteratively updates $\bm\theta$ within $U(\bm\theta)$ until reaching the predetermined number of steps, ultimately providing the approximation of the maximum point. \textbf{(b)}~Detailed optimization steps of BO on the given Maximum Cut graph instance. In the 10th step, the acquisition function $\mathrm{UCB}_{10}(\bm\theta)$ are calculated based on the accumulated points $\mathcal{A}_9$ from the previous 9 steps. The next point $\bm\theta_{10}$ is selected by maximizing $\mathrm{UCB}_{10}(\bm\theta)$, and the current approximation of the maximum point $\bm\theta_{10}^+$ is updated to the best point selected in the previous 10 steps. Finally, after 20 steps, the final approximation of the maximum point $\bm\theta_{20}^+$ is returned. \textbf{(c)}~The estimation of the function value at the approximation of the maximum point $y(\bm\theta_t^+)$ as a function of the step $t$.}
\label{fig:QAOA_BO}
\end{figure*}

\subsection{Acquisition Function}
In the $t$-th step of BO, the acquisition function utilizes $\mathcal{S}_{t-1}$ to guide the search towards the next point $\bm\theta_{t}$, aiming to converge to the global maximum point $\bm\theta^{*}$ of $f(\bm\theta)$. This procedure is accomplished by maximizing the acquisition function over the domain $\mathcal{D}$. Specifically, the design of the acquisition function should consider both exploration (sampling in regions of high uncertainty) and exploitation (sampling in regions likely to yield high function values). 
The upper confidence bound $\mathrm{UCB}_t(\bm\theta)$, which is a commonly employed acquisition function, is defined as 
\begin{align}\label{eq:UCB}
  \mathrm{UCB}_t(\bm\theta)=\mu_{t-1}(\bm\theta)+\sqrt{\eta_t}\sigma_{t-1}(\bm\theta),
\end{align} and the next point \(\bm\theta_t\) is selected as 
\begin{align}
\label{eq:Next}
\bm\theta_t=\arg\max_{\bm \theta \in\mathcal{D}} \mathrm{UCB}_t(\bm\theta),
\end{align}
where \(\mu_{t-1}(\bm\theta)\) and \(\sigma_{t-1}(\bm\theta)\) denote the posterior mean function and the posterior standard deviation respectively, as defined in Eq.~\ref{eq:Posterior}, and \(\eta_t\geq 0\) represents a time-dependent scaling parameter. Subsequently, the accumulated observations are updated as $\mathcal{S}_t=\{(\bm\theta_{1}, y(\bm\theta_{1})), \cdots, (\bm\theta_{t},y(\bm\theta_{t}))\}$, and the posterior distribution for $f(\bm\theta)$ is updated based on $\mathcal{S}_t$. 

\vspace{10pt}
The aforementioned process is repeated for a predetermined number of steps $T$. The best point selected in the previous $T$ steps represents the approximation of the maximum point $\bm\theta_T^+$. Specifically, 
\begin{equation}
    \bm\theta_T^+=\arg \max_{\bm\theta\in\mathcal{A}_T} y(\bm\theta),
\end{equation}
where \(\mathcal{A}_{T}=\{\bm\theta_1,\cdots,\bm\theta_{T}\}\) represents the accumulated points from the previous $T$ steps. Figure~\ref{fig:QAOA_BO} illustrates the numerical performance of QAOA for the Maximum Cut problem on a specific graph instance using the Bayesian approach. In conclusion, the Gaussian process is widely preferred as the statistical model in BO due to its flexibility and capacity to model complex functions. It offers a powerful framework for modeling the objective function by capturing both the mean and uncertainty associated with observations of the objective function. This enables BO to make informed decisions regarding the location of the next point. Subsequently, the upper confidence bound balances the trade-off between exploration and exploitation and selects the next point based on the current knowledge provided by the Gaussian process. In summary, the combination of the Gaussian process as the statistical model and the 
upper confidence bound as the acquisition function constitutes the core of the BO framework, enabling efficient global optimization in the absence of gradient information. 

\begin{figure*}[htpb]
\centering
\subfigure[]{
\begin{minipage}{16.8cm} 
\label{fig:Main_Results_1}
\includegraphics[width=0.99\textwidth]
{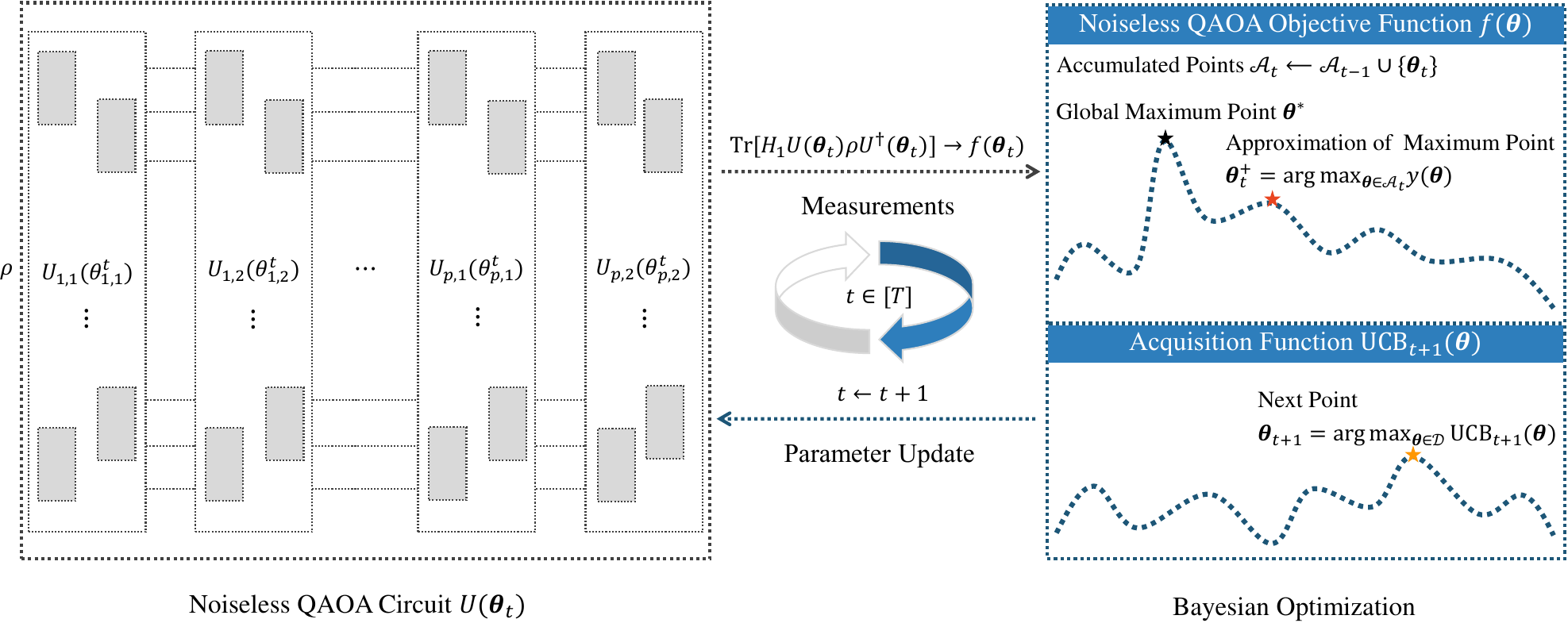} 
\end{minipage}
}
\subfigure[]{
\begin{minipage}{16.8cm}
\label{fig:Main_Results_2}
\includegraphics[width=0.99\textwidth]{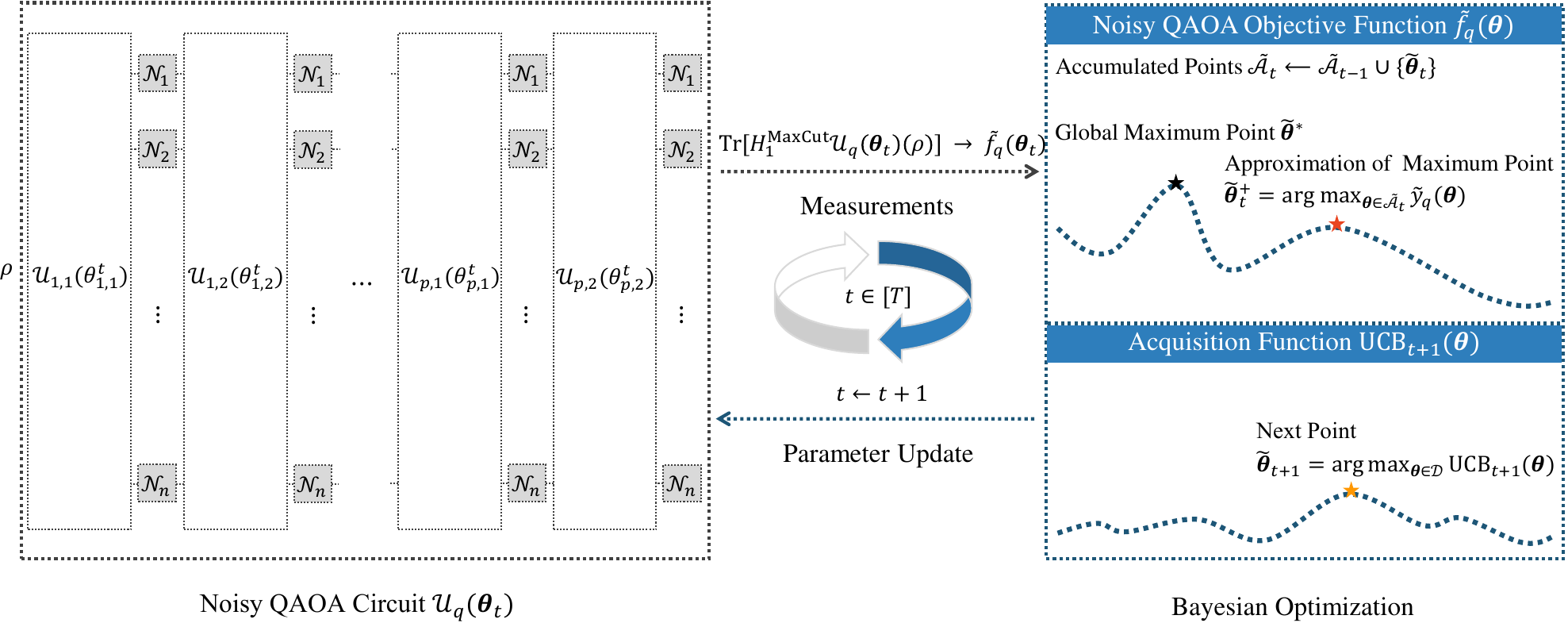} 
\end{minipage}
}
\caption{\textbf{Scenarios to investigate the trainability of QAOA.} \textbf{(a)}~Optimization of the $n$-qubit noiseless QAOA circuit using the Bayesian approach. In the $t$-th step of the quantum subroutine, the noiseless QAOA circuit $U(\bm\theta_t)$ prepares parameterized quantum states, where either the left or right slice of each block in $U(\bm\theta_t)$ forms a local 1-design. Next, the noiseless QAOA objective function $f(\bm\theta_t)$ is estimated through the fixed number of measurements $M$, yielding the estimation $y(\bm\theta_t)$. In the classical subroutine, BO utilizes the acquisition function $\mathrm{UCB}_{t+1}(\bm\theta)$ to select the next point $\bm\theta_{t+1}$ based on the current knowledge provided by the Gaussian process. Afterwards, the variational parameters are updated. This process is repeated for a predetermined number of steps $T$, and the best point selected in the previous $T$ steps represents the approximation of the maximum point $\bm\theta_T^+$. \textbf{(b)}~Optimization of the $n$-qubit noisy QAOA circuit using the Bayesian approach. In the $t$-th step of the quantum subroutine, the noisy QAOA circuit $\mathcal{U}_{q}(\bm\theta_t)$ prepares parameterized quantum states, where a noise channel $\mathcal{N}=\bigotimes_{i=1}^n\mathcal{N}_i$ exists between any two blocks in the circuit, and $\mathcal{N}_i$ represents a local Pauli channel acting on the $i$-th qubit. Next, the noisy QAOA objective function $\tilde f_q(\bm\theta_t)$ is estimated through the fixed number of measurements $M$, yielding the estimation $\tilde y_q(\bm\theta_t)$. In the classical subroutine, BO iteratively updates the variational parameters using the same optimization process for $T$ steps and eventually returns the approximation of the maximum point $\tilde{\bm\theta}_T^+$. The presence of noise may result in a flatter landscape with fewer local optima for $\tilde f_q(\bm\theta)$ compared to $f(\bm\theta)$.}
\end{figure*}

\section{Main Results I: Analyzing the trainability of the noiseless QAOA circuit}\label{sec:Theoretical Analysis 1} 
Our main focus is to theoretically investigate the trainability of the $n$-qubit noiseless QAOA circuit $U(\bm\theta)$ using the Bayesian approach. The optimization error $r_T$ after $T$ steps of executing BO is defined as the difference in function values between the global maximum point $\bm\theta^{*}$ and the approximation of the maximum point $\bm\theta_T^+$ in the previous $T$ steps. It is given by
\begin{align}\label{eq:ERROR 1}
r_T=f(\bm\theta ^*)-f(\bm\theta^+_T),
\end{align}
where \(\bm\theta_T^+=\arg \max_{\bm\theta\in\mathcal{A}_T} f(\bm\theta)\) with the accumulated points \(\mathcal{A}_{T}=\{\bm\theta_1,\cdots,\bm\theta_{T}\}\)\footnote{
We assume by default that the estimations $\{y(\bm\theta)\}_{\bm\theta\in\mathcal{A}_T}$ from the previous $T$ steps contain sufficient information about the noiseless QAOA objective function values $\{f(\bm\theta)\}_{\bm\theta\in\mathcal{A}_T}$ with the accumulated points $\mathcal{A}_T=\{\bm\theta_1,\ldots,\bm\theta_T\}$. Therefore, it is reasonable to define the approximation of the maximum point as $\bm\theta_T^+=\arg \max_{\bm\theta\in\mathcal{A}_T} f(\bm\theta)$, even though it was originally defined as $\bm\theta_T^+=\arg \max_{\bm\theta\in\mathcal{A}_T} y(\bm\theta)$.}. We aim to determine the effective depth of $U(\bm\theta)$, ensuring that $r_T\leq \epsilon$ after $T={\rm poly}(n)$ steps, where $\epsilon$ represents a constant threshold. It is worth noting that in the context of QAOA, the parameter dimension and the circuit depth exhibit similar magnitudes. Hence, we can directly explore the effective parameter dimension $p$ of $U(\bm\theta)$ in subsequent analysis. In this work, we adopt the following widely accepted assumption.

\begin{assumption}[\cite{cerezo2021cost,harrow2009random}]
Given an $n$-qubit noiseless QAOA circuit 
\begin{equation}\label{eq:noiseless QAOA circuit}
U(\bm\theta)=\prod_{j=1}^p\prod_{l=1}^2 U_{j,l}(\theta_{j,l}),    
\end{equation}
each block $U_{j,l}(\theta_{j,l})=U^{(j,l)}_{+}U^{(j,l)}_{-}$ for $(j,l)\in[p]\times[2]$, where $U^{(j,l)}_{-}$ is independent to $U^{(j,l)}_{+}$, and at least one of them forms a local 1-design.
\label{assump:1-design}
\end{assumption}

The scenario mentioned above for investigating the trainability of $U(\bm\theta)$ using the Bayesian approach is described in detail in Figure~\ref{fig:Main_Results_1}. Assuming that Assumption~\ref{assump:1-design} holds, we first explore the Lipschitz continuity of the corresponding noiseless QAOA objective function $f(\bm\theta)$. Additionally, we establish a theoretical limit on $p$ that ensures achieving $r_T\leq \epsilon$ within $T={\rm poly}(n)$ steps. The following sections provide a comprehensive introduction.

\subsection{Continuity Property of the Noiseless QAOA Objective Function}
Now, we will show that Assumption~\ref{assump:1-design} results in a quantum analog of the Lipschitz continuity property about the noiseless QAOA objective function $f(\bm\theta)$.

\begin{lemma}\label{lemma:Lipschitz Continuity 1}
Assuming that Assumption~\ref{assump:1-design} holds, let $f(\bm\theta): \mathcal{D}=[0,2\pi]^{2p}\mapsto\mathbb{R}$ be the noiseless QAOA objective function~(Eq.~\ref{eq:Objective Function 1}). Given a failure probability $\delta\in(0,1)$, for any $\bm\theta,\bm\theta^{\prime}\in\mathcal{D}$, the relationship
\begin{equation}
\left|f(\bm\theta)-f(\bm\theta^{\prime})\right|\leq \sqrt{\mathbb{V}_{\bm\theta}[\partial_a f(\bm\theta)]/\delta}\Vert\bm\theta-\bm\theta^{\prime}\Vert_1 
\end{equation}
is valid with a success probability of at least \( 1-\delta\), where $\mathbb{V}_{\bm\theta}[\partial_a f(\bm\theta)]$ is the variance of the partial derivative $\partial_a f(\bm\theta)$ with index $a=\arg \max_{j\in[2p]}(\sup_{\bm\theta\in
\mathcal{D}}\left|\partial_j f(\bm\theta)\right|)$.
\end{lemma}

\begin{proof}[Proof Sketch] 
Given Assumption~\ref{assump:1-design}, we first show that $\mathbb{E}_{\bm\theta}[\partial_jf(\bm\theta)]=0$ for arbitrary $j\in [2p]$. Then for any $j\in[2p]$ and $\bm\theta\in\mathcal{D}$, it is proved that $\left|\partial_jf(\bm\theta)\right|\leq\sqrt{\mathbb{V}_{\bm\theta}[\partial_a f(\bm\theta)]/\delta}$ with the success probability of at least $1-\delta$. Here, $\mathbb{V}_{\bm\theta}[\partial_a f(\bm\theta)]$ is the variance of the partial derivative $\partial_af(\bm\theta)$ with the index $a=\arg \max_{j\in[2p]}(\sup_{\bm\theta\in
\mathcal{D}}\left|\partial_jf(\bm\theta)\right|)$. Furthermore, by Lagrange’s Mean Value Theorem~\cite{sohrab2003basic}, we obtain $\left|f(\bm\theta)-f(\bm\theta^{\prime})\right| \leq
\max_{j\in[2p]}(\sup_{\bm\theta\in
\mathcal{D}}\left|\partial_j f(\bm\theta)\right|)\Vert\bm\theta-\bm\theta^{\prime}\Vert_1$ for any $\bm\theta,\bm\theta^{\prime}\in\mathcal{D}$.  
Combining these two results, we finally complete the proof of Lemma~\ref{lemma:Lipschitz Continuity 1}, and proof details can be found in Appendix~\ref{appendix:Lemma1}.   
\end{proof}

In addition to Assumption~\ref{assump:1-design}, we provide the following remark based on the differentiable property of $f(\bm\theta)$.

\begin{remark}
    Given the differentiable property of the noiseless QAOA objective function $f(\bm\theta)$~\cite{bouland2019complexity, wu2023complexity}, we can consider it as a sample drawn from a Gaussian process with the Matern prior covariance function $k_{\rm Matern-\nu}(\bm\theta,\bm\theta^{\prime})$~(Eq.~\ref{eq:Matern}), as this Gaussian process allows us to model high-order differentiable functions~\cite{rasmussen2006gaussian,kanagawa2018gaussian}.
\end{remark}

\subsection{Effective Parameter Dimension of the Noiseless QAOA Circuit}
Assuming that Assumption~\ref{assump:1-design} holds and using the result of Lemma~\ref{lemma:Lipschitz Continuity 1}, we establish a theoretical limit on parameter dimension $p$ through the lens of the Bayesian approach, such that $r_T$ can be upper bounded by a constant threshold $\epsilon$ within $T={\rm poly}(n)$ steps. 

\begin{theorem}[Informal]
Given a constant threshold $\epsilon$ and 
an $n$-qubit noiseless QAOA circuit $U(\bm\theta)$~(Eq.~\ref{eq:noiseless QAOA circuit}) that satisfies Assumption~\ref{assump:1-design}, run Algorithm~\ref{alg:BO for QAOA} for $T={\rm poly}(n^{1/\epsilon^2})$ steps, where a predefined scaling parameter $\eta_t$ for the acquisition function $\mathrm{UCB}_t(\bm\theta)$~(Eq.~\ref{eq:UCB}) is used in each step $t$. If the parameter dimension
\begin{equation}
p\leq\tilde{\mathcal{O}}\left(\sqrt{\log n}\right),
\end{equation}
then the optimization error $r_T$~(Eq.~\ref{eq:ERROR 1}) satisfies $r_T\leq\epsilon$ with high success probability.
\label{theorem:Theorem 1}
\end{theorem}

\begin{proof}[Proof Sketch] 
In order to determine the effective $p$ that guarantees $r_T\leq \epsilon$, we first establish that $r_T$ is upper bounded by $\frac{1}{T}\sum_{t=1}^T\left(f(\bm\theta ^*)-f(\bm\theta_t)\right)$. Here, $\bm\theta^{*}$ represents the global maximum point, and $\bm\theta_t$ denotes the next point selected in each step $t$. It is evident that the condition $\frac{1}{T}\sum_{t=1}^T\left(f(\bm\theta ^*)-f(\bm\theta_t)\right)\leq \epsilon$ is sufficient to deduce the result $r_T\leq\epsilon$. Hence, by ensuring that the upper bound on $\frac{1}{T}\sum_{t=1}^T\left(f(\bm\theta ^*)-f(\bm\theta_t)\right)$ is no greater than $\epsilon$, we can determine the effective $p$ that guarantees $r_T\leq \epsilon$. Subsequently, we utilize the continuity property of $f(\bm\theta)$~(Lemma~\ref{lemma:Lipschitz Continuity 1}) to establish an upper bound on $\frac{1}{T}\sum_{t=1}^T\left(f(\bm\theta ^*)-f(\bm\theta_t)\right)$. The formal statement and corresponding proof details are provided in Appendix~\ref{appendix:Theorem1}.
\end{proof}
 
In summary, we consider the $n$-qubit noiseless QAOA objective function $f(\bm\theta)$ as a sample drawn from a Gaussian process with the Matern covariance function $k_{\rm Matern-\nu}(\bm\theta,\bm\theta^{\prime})$, leveraging its high-order differentiable property. We then investigate the trainability of the corresponding noiseless QAOA circuit $U(\bm\theta)$ through this lens. Based on Assumption~\ref{assump:1-design} that either the left or right slice of each block in $U(\bm\theta)$ forms a local 1-design, we demonstrate that $U(\bm\theta)$ with a parameter dimension $p$ of $\tilde{\mathcal{O}}(\sqrt{\log n})$ can be trained efficiently using the Bayesian approach. 

\section{Main Results II: Analyzing the trainability of the noisy QAOA circuit}\label{sec:Theoretical Analysis 2}
After exploring the trainability of the noiseless QAOA circuit $U(\bm\theta)$ through the Bayesian approach, we will proceed to investigate its theoretical performance in a practical scenario. This scenario involves the Maximum Cut problem on an unweighted regular graph, where $U(\bm\theta)$ is affected by local Pauli channels. For the sake of clarity, we begin by presenting the definitions of the Maximum Cut problem and the local Pauli channel.

\begin{definition}[Maximum Cut problem]
Considering an unweighted $d$-regular graph $G=(V, E)$ with the vertices set $V=\{v_1,\cdots,v_n\}$ and the edges set $E=\{e_{i,j}\}$, the Maximum Cut problem aims at dividing all vertices into two disjoint sets such that maximizing the number of edges that connect the two sets. In the context of QAOA, the problem-oriented Hamiltonian $H_1^{\rm {MaxCut}}$ is defined as 
\begin{equation}\label{eq:H1_MaxCut}
  H_1^{\rm {MaxCut}}= \frac{1}{2}\sum_{e_{i,j}\in E}(\mathbb{I}^{\otimes n}-\sigma_{i}^{z}\sigma_{j}^{z}).
\end{equation}
\end{definition}

\begin{definition}[Local Pauli channel]
Let $\mathcal{N}_i$ denote a local Pauli channel and the action of $\mathcal{N}_i$ is random local Pauli operators $P$ acting on the $i$-th qubit according to a specific probability distribution $\{q(P)\}$, where $P\in\{\mathbb{I}, \sigma^x,\sigma^y,\sigma^z\}$. Specifically, the action of $\mathcal{N}_i$ is given by 
\begin{equation}\label{eq:Local Pauli Noise Channel}
    \mathcal{N}_i(\cdot)=\sum\limits_{P\in\{\mathbb{I},\sigma^x,\sigma^y, \sigma^z\}}\ q(P) P(\cdot) P^{\dagger},
\end{equation}
where $q(P)\in(0,1)$ and $\sum_{P\in\{\mathbb{I},\sigma^x,\sigma^y, \sigma^z\}}q(P)=1$. The noise strength in this model is represented by a single parameter $q=\max_{P\in\{\sigma^x,\sigma^y, \sigma^z\}}q(P)$.
\end{definition}

Due to imperfections in quantum devices, we assume that each quantum gate is affected by a $q$-strength local Pauli channel, and the effects of these noises are postponed until the end of each block in $U(\bm\theta)$. This assumption is reasonable, as it has been employed in Ref~\cite{wang2021noise} and demonstrated to hold true in Clifford circuits~\cite{quek2022exponentially}. Below is a detailed and precise description of this assumption.

\begin{assumption}\label{assump:local Pauli noise channel}
Given the $q$-strength local Pauli channel $\mathcal{N}_i$~(Eq.~\ref{eq:Local Pauli Noise Channel}) which is gate-independent and time-invariant, the $n$-qubit noisy QAOA circuit is given by 
\begin{align}\label{eq:noisy QAOA circuit}
\mathcal{U}_{q}(\bm\theta)=\bigcirc_{j=1}^p\bigcirc_{l=1}^2\left(\mathcal{N}\circ \mathcal{U}_{j,l}(\theta_{j,l})\right),
\end{align}
where $\mathcal{N}=\bigotimes_{i=1}^n\mathcal{N}_i$ is the noise channel and $\mathcal{U}_{j,l}(\theta_{j,l})$ is the channel that implements the unitary $U_{j,l}(\theta_{j,l})$ for $(j,l)\in[p]\times[2]$.
\end{assumption}

Assuming that Assumption~\ref{assump:local Pauli noise channel} holds, the $n$-qubit noisy QAOA objective function $\tilde {f}_{q}(\bm\theta)$ with $q$-strength local Pauli channels is given by
\begin{equation}\label{eq:Objective Function 2}
    \tilde {f}_q(\bm\theta)=\mathrm {Tr}\left[H_1^{\rm {MaxCut}}\mathcal{U}_{q}(\bm\theta)(\rho)\right],
\end{equation}
where $H_1^{\rm {MaxCut}}$ is the problem-oriented Hamiltonian about the Maximum Cut problem~(Eq.~\ref{eq:H1_MaxCut}), $\mathcal{U}_{q}(\bm\theta)$ is the noisy QAOA circuit~(Eq.~\ref{eq:noisy QAOA circuit}) and $\rho=(|+\rangle\langle+|)^{\otimes n}$ is the initial state. Now, the optimization error $\tilde r_T$ after $T$ steps of executing BO is defined as the difference in function values between the global maximum point  $\tilde{\bm\theta}^{*}$ and the approximation of the maximum point $\tilde{\bm\theta}_T^+$ in the previous $T$ steps. Specifically, 
\begin{align}\label{eq:ERROR 2}
    \tilde r_T=\tilde f_{q}(\tilde{\bm\theta} ^*)-\tilde f_{q}(\tilde{\bm\theta}^+_T),
\end{align}
where $\tilde{\bm\theta}_T^+=\arg \max_{\bm\theta\in\mathcal{\tilde{A}}_T} \tilde f_{q}(\bm\theta)$ with the accumulated points \(\mathcal{\tilde{A}}_{T}=\{\tilde{\bm\theta}_1,\cdots,\tilde{\bm\theta}_{T}\}\). Figure~\ref{fig:Main_Results_2} provides a detailed description of the scenario mentioned above for investigating the trainability of $\mathcal{U}_{q}(\bm\theta)$ using the Bayesian approach. In the following sections, we first explore the Lipschitz continuity of $\tilde {f}_q(\bm\theta)$. As previously mentioned, the parameter dimension and the circuit depth share similar magnitudes in the context of QAOA. Hence, we explore the effective parameter dimension $p$ of $\mathcal{U}_{q}(\bm\theta)$ directly and establish a theoretical limit on $p$ that ensures achieving $\tilde r_T\leq \epsilon$ within $T={\rm poly}(n)$ steps. Here, $\epsilon$ denotes a constant threshold. 

\subsection{Continuity Property of the Noisy QAOA Objective Function}
Now, we will show that Assumption~\ref{assump:local Pauli noise channel} results in a quantum analog of the Lipschitz continuity property about the noisy QAOA objective function $\tilde f_q(\bm\theta)$.

\begin{lemma}\label{lemma:Lipschitz Continuity 2}
Assuming that Assumption~\ref{assump:local Pauli noise channel} holds and considering the Maximum Cut problem on an unweighted $d$-regular graph with $n$ vertices, let $\tilde{f}_{q}(\bm\theta):\mathcal{D}=[0,2\pi]^{2p}\mapsto\mathbb{R}$ be the noisy QAOA objective function with $q$-strength local Pauli channels~(Eq.~\ref{eq:Objective Function 2}). For any $\bm\theta,\bm\theta^{\prime}\in\mathcal{D}$, the relationship 
\begin{equation}
\left|\tilde f_q(\bm\theta)-\tilde f_q(\bm\theta^{\prime})\right|\leq d^{3}n^{7/2}q^{(d+1)p}\Vert\bm\theta-\bm\theta^{\prime}\Vert_1  
\end{equation}
is valid, where the noise strength $q\in(0,1)$.
\end{lemma}

The proof sketch of this lemma is similar to Lemma~\ref{lemma:Lipschitz Continuity 1} and the detailed proof can be found in Appendix~\ref{appendix:Lemma2}. In addition to Assumption~\ref{assump:local Pauli noise channel}, we provide the following remark based on the differentiable property of $\tilde f_q(\bm\theta)$.

\begin{remark}
Given the differentiable property of the noisy QAOA objective function $\tilde {f}_{q}(\bm\theta)$~\cite{fontana2023classical}, we can consider it as a sample drawn from a Gaussian process with the Matern prior covariance function $k_{\rm Matern-\nu}(\bm\theta,\bm\theta^{\prime})$~(Eq.~\ref{eq:Matern}), as this Gaussian process allows us to model high-order differentiable functions~\cite{rasmussen2006gaussian,kanagawa2018gaussian}.
\end{remark}

\subsection{Effective Parameter Dimension of the Noisy QAOA Circuit}
Assuming that Assumption~\ref{assump:local Pauli noise channel} holds and using the result of Lemma~\ref{lemma:Lipschitz Continuity 2}, we establish a theoretical limit on the parameter dimension $p$ through the lens of the Bayesian
approach, such that $\tilde r_T$ can be upper bounded by a constant threshold $\epsilon$ within $T={\rm poly}(n)$ steps.

\begin{theorem}[Informal]
\label{theorem:Theorem 2}
Consider the Maximum Cut problem on an unweighted $d$-regular graph with $n$ vertices, where $d$ is a constant. Given a constant threshold $\epsilon$ and  
a noisy QAOA circuit $\mathcal{U}_{q}(\bm\theta)$ with $q$-strength local Pauli channels~(Eq.~\ref{eq:noisy QAOA circuit}) that satisfies Assumption~\ref{assump:local Pauli noise channel}, run Algorithm~\ref{alg:BO for QAOA} for $T={\rm poly}(n^{1/\epsilon^2})$ steps, where a predefined scaling parameter $\eta_t$ for the acquisition function $\mathrm{UCB}_t(\bm\theta)$~(Eq.~\ref{eq:UCB}) is used in each step $t$. Under the condition where the noise strength $q$ spans $1/{\rm poly} (n)$ to $1/n^{1/\sqrt{\log n}}$, if the parameter dimension
 \begin{equation}
    p\leq\mathcal{O}\left(\log n/\log(1/q)\right),
\end{equation}     
then the optimization error $\tilde r_T$~(Eq.~\ref{eq:ERROR 2}) satisfies $\tilde r_T\leq\epsilon$ with high success probability.  
\end{theorem}

The proof sketch of this theorem is similar to Theorem~\ref{theorem:Theorem 1}. The formal statement and corresponding proof details are provided in Appendix~\ref{appendix:Theorem2}. Following our previous lens of the Bayesian approach, we consider the $n$-qubit noisy QAOA objective function $\tilde f_q(\bm\theta)$ as a sample drawn from a Gaussian process with the Matern covariance function $k_{\rm Matern-\nu}(\bm\theta,\bm\theta^{\prime})$. 
Using this framework, we investigate the trainability of the corresponding noisy QAOA circuit $\mathcal{U}_{q}(\bm\theta)$ within a practical scenario concerning the Maximum Cut problem on an unweighted regular graph. Based on Assumption~\ref{assump:local Pauli noise channel}, we show that if each quantum gate is affected by a $q$-strength local Pauli channel with the noise strength range of $1/{\rm poly} (n)$ to $1/n^{1/\sqrt{\log n}}$, $\mathcal{U}_{q}(\bm\theta)$ with a parameter dimension $p$ of $\mathcal{O}\left(\log n/\log(1/q)\right)$ can also be trained efficiently. For a more intuitive description of the noise strength range described above, we focus on near-term quantum devices with 50-100 qubits~\cite{preskill2018quantum}. In this case, $1/n^{1/\sqrt{\log n}}$ is only slightly larger than 0.1. This suggests that this range corresponds to the actual noise levels in near-term quantum devices and holds practical significance. Notably, $\mathcal{O}\left(\log n/\log(1/q)\right)$ grows as $q$ increases, implying the growth in the depth of $\mathcal{U}_{q}(\bm\theta)$ that can be effectively trained as the noise strength increases. This finding implies that a specific strength noise may obscure certain local optima, thereby improving the trainability.

\section{Discussion}
In order to emphasize our contributions in characterizing the trainability of QAOA, a topic of utmost importance in the NISQ era, we provide a brief discussion on the connections between our work and related studies in terms of fundamental assumptions and obtained results. Earlier studies on this topic have investigated it from two key perspectives:

\begin{enumerate}
    \item Theoretical convergence performance of the QAOA objective function under the strongly-convex assumption or satisfying Polyak-Lojasiewicz inequality~\cite{harrow2021low,sweke2020stochastic}.
    \item Estimation complexity of the gradient for the QAOA objective function investigated by analyzing the existence of the barren plateau phenomenon~\cite{mcclean2018barren,cerezo2021cost,wang2021noise}.
\end{enumerate}

The theoretical convergence performance of the QAOA objective function within the context of SGD was examined in simplified settings, assuming either strong convexity or the satisfaction of Polyak-Lojasiewicz inequality~\cite{harrow2021low,sweke2020stochastic}. However, in realistic scenarios, the landscape of the QAOA objective function is generally non-convex and contains numerous local optima. Therefore, these studies focus on the scenario where the QAOA objective function, given suitable initial parameters, resides in a convex vicinity containing a global or local optimum. In general, finding appropriate initial parameters is not always easy. The theoretical convergence performance of the QAOA objective function with the non-convex landscape remains unclear when utilizing random initial parameters. 
Our work investigates the trainability of the QAOA objective function within the context of BO, a gradient-free global optimization method. In contrast to gradient-based optimization methods, BO demonstrates its favorable performance in handling the non-convex objective function. We no longer rely on explicit assumptions regarding the landscape of the QAOA objective function. Instead, we leverage the high-order differentiable property of the QAOA objective function and consider it as a sample drawn from a Gaussian process with the Matern covariance function, regardless of the presence of noise. Additionally, there are several widely accepted assumptions. For instance, either the left or right slice of each block in the noiseless QAOA circuit is assumed to form a local 1-design. For the Maximum Cut problem on an unweighted regular graph, we make a commonly employed assumption that each quantum gate is affected by a local Pauli channel, and the effects of these noises are postponed until the end of each block in the QAOA circuit. Investigating the trainability of QAOA through the Bayesian approach in these more realistic scenarios holds substantial practical significance.

The estimation complexity of the gradient for the QAOA objective function was investigated by analyzing the existence of the barren plateau phenomenon~\cite{mcclean2018barren}. In this phenomenon, the probability that the gradient along any reasonable direction is non-zero to some fixed precision is exponentially small as a function of the number of qubits, rendering the estimation complexity of the gradient unacceptable. Refs~\cite{cerezo2021cost,wang2021noise} examine the causes of the barren plateau phenomenon in both noiseless and noisy scenarios, providing insights into situations where this phenomenon can be mitigated. For instance, Ref~\cite{cerezo2021cost} suggested that when each block in the $n$-qubit noiseless QAOA circuit forms a local 2-design, the corresponding QAOA objective function with a local problem-oriented Hamiltonian leads to at worst a polynomially vanishing gradient as long as the depth of the circuit is $\mathcal{O}(\log n)$. Meanwhile, similar findings also exist in the noisy scenario~\cite{wang2021noise}. 
In such situations, although the estimation complexity of the gradient may be acceptable, it does not necessarily guarantee the effective training of the corresponding QAOA circuit. 
Our work further investigates the trainability of QAOA in the aforementioned situations through the lens of the Bayesian approach. When either the left or right slice of each block in the $n$-qubit noiseless QAOA circuit forms a local 1-design, we demonstrate that the circuit with a depth of $\tilde{\mathcal{O}}(\sqrt{\log n})$ can be trained efficiently. Moreover, we show that if each quantum gate in the $n$-qubit QAOA circuit is affected by a $q$-strength local Pauli channel with the noise strength range of $1/{\rm poly} (n)$ to 0.1, the noisy QAOA circuit with a depth of $\mathcal{O}\left(\log n/\log(1/q)\right)$ can also be trained efficiently for solving the Maximum Cut problem on an unweighted regular graph.

\section{Conclusion}
In summary, we propose a novel lens for comprehending the QAOA objective function. The QAOA objective function exhibits the properties of a high-order differentiable function in both noiseless and noisy scenarios. Moreover, a Gaussian process with the Matern covariance function has the remarkable ability to generate high-order differentiable functions. As a result, regardless of the presence of noise, the QAOA objective function can be considered as a sample drawn from a Gaussian process with the Matern covariance function. This paradigm shift allows us to investigate the trainability of the QAOA circuit using information theory in the context of BO, without relying on explicit assumptions regarding the landscape of the corresponding QAOA objective function. We demonstrate the $n$-qubit noiseless QAOA circuit with a depth of $\tilde{\mathcal{O}}(\sqrt{\log n})$ can be trained efficiently, based on the widely accepted assumption that either the left or right slice of each block in the circuit forms a local 1-design. Additionally, we show that if each quantum gate in the $n$-qubit QAOA circuit is affected by a $q$-strength local Pauli channel with the noise strength range of $1/{\rm poly} (n)$ to 0.1, the noisy QAOA circuit with a depth of $\mathcal{O}\left(\log n/\log(1/q)\right)$ can also be trained efficiently for solving the Maximum Cut problem on an unweighted regular graph. It should be noted that our work focuses on the trainability of the QAOA circuit. To solve constraint satisfaction problems, this depth of the QAOA circuit may not provide optimal solutions. In conclusion, these two findings represent significant progress in understanding the complexity of the iterations involved in optimizing the variational parameters within the QAOA circuit. Our results offer valuable insights into the theoretical performance of quantum optimization algorithms in the NISQ era.

This work not only opens up a new avenue for further investigation but also identifies an additional worthwhile topic to explore. We establish an upper bound on the optimization error and derive a theoretical limit on the depth of the QAOA circuit. This limit ensures that the optimization error converges to a constant level as the number of optimization iterations scales polynomially with the number of qubits in the context of BO. To investigate the scenario where effective training of the QAOA circuit becomes unattainable, i.e., the optimization error fails to converge within a constant level as the number of optimization iterations scales polynomially with the number of qubits, a lower bound on the optimization error is necessary. Therefore, exploring this lower bound in the training of the QAOA circuit through the Bayesian approach is an intriguing research topic to pursue. This investigation will lead to a more comprehensive understanding of the complexity of the iterations associated with optimizing the variational parameters within the QAOA circuit.
 
\section*{Acknowledgments}
This work is supported by National Natural Science Foundation of China (Grant Nos.~62272056, 61976024, 62371069), China Scholarship Council (Grant No.~202006470011) and UWA-CSC HDR
Top-Up Scholarship (Grant No.~230808000954).

\bibliography{Trainability_Analysis}

\clearpage
\widetext
\widetext
\appendix

\section{Related Definitions}
This section presents background information on the Matern covariance function, differential entropy, and information gain. 

\subsection{Matern Covariance Function}\label{appendix:Matern}
The Matern covariance function, widely used in BO, is defined as
\begin{equation}\label{eq:Matern}
    k_{\mathrm{Matern}-\nu}(\bm\theta,\bm\theta^{\prime})=\frac{1}{\Gamma(\nu)2^{\nu-1}}\left(\frac{\sqrt{2\nu}d}{l}\right)^{\nu}B_{\nu}\left(\frac{\sqrt{2\nu}d}{l}\right),
\end{equation}
where $l>0$, $d=\Vert\bm\theta-\bm\theta^{\prime}\Vert_2$ represents the Euclidean distance between \(\bm\theta\) and \(\bm\theta^{\prime}\), \(\nu>0\) denotes the smoothness parameter, $\Gamma(\cdot)$ represents the gamma function, and $B_{\nu}(\cdot)$ denotes the modified Bessel function of the second kind. Varying \(\nu\) determines the smoothness of samples drawn from a Gaussian process with this covariance function. Smaller values of \(\nu\) correspond to rougher samples. Additionally, these samples are $\lceil \nu \rceil-1$ times continuously differentiable~\cite{rasmussen2006gaussian}. Figure~\ref{fig:samples} illustrates samples drawn from a Gaussian process with this covariance function using different values of $\nu$. 

\begin{figure*}[htpb]
\centering
\includegraphics[width=0.99\textwidth]{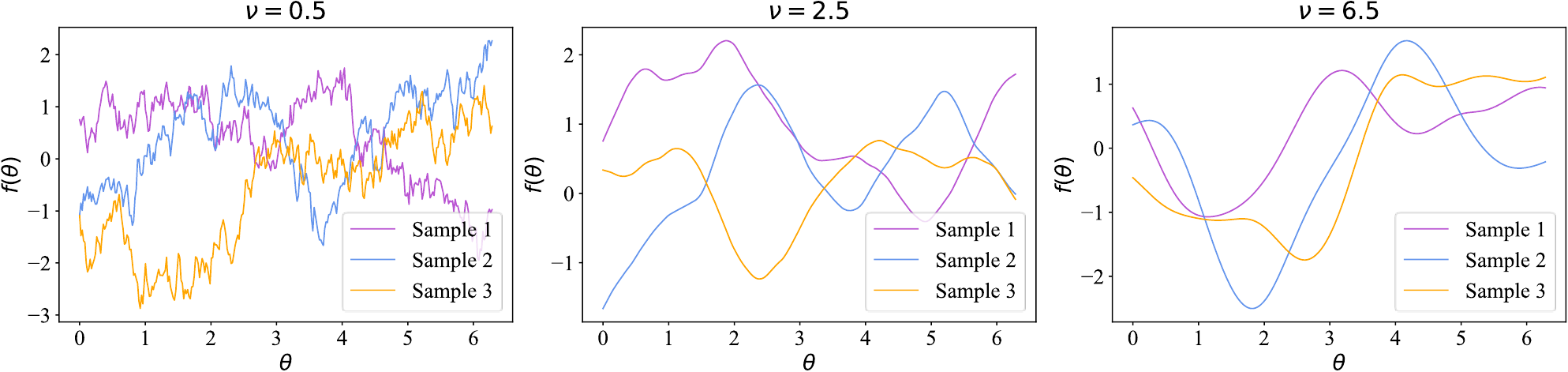} 
\caption{Samples drawn from a Gaussian process with the Matern covariance function $k_{\mathrm{Matern}-\nu}(\bm\theta,\bm\theta^{\prime})$ using smoothness parameters $\nu$ from $\nu=0.5$ to $\nu=6.5$.}
\label{fig:samples}
\end{figure*}

\subsection{Differential Entropy} 
Let $X$ be a random variable with a probability density function $q$ whose support is a set $\mathcal{X}$. The differential entropy $H(X)$ is defined as
\begin{equation}
    \mathrm{H}[X]=\mathbb{E}[-\log(q(X))]=\int_{\mathcal{X}}q(x)\log q(x) dx.
\end{equation}

Specifically, the differential entropy of a multivariate Gaussian random variable $X_{\rm Gaussian}$ with distribution $N(\bm \mu,\bm K)$ is expressed as 
\begin{equation}\label{eq:entropy}
    \mathrm{H}[X_{\rm Gaussian}]=\frac{1}{2}\log(\det(2\pi e\bm K)),
\end{equation}
where $\bm\mu$ denotes the mean vector and $\bm K$ represents the covariance matrix.

\subsection{Information Gain}
Let $\mathcal{S}_T=\{(\bm\theta_{1}, y(\bm\theta_{1})), \cdots, (\bm\theta_{T},y(\bm\theta_{T}))\}$ be $T$ accumulated observations about the function $f(\bm\theta)$, where $y(\bm\theta_t)$ denotes the estimation of $f(\bm\theta_t)$ for $t\in[T]$. The informativeness of $\mathcal{S}_T$ regarding $f(\bm\theta)$ is quantified by the information gain $g_T$, which is the mutual information~\cite{shannon1948mathematical} between $\bm {y}_{T}=[y(\bm\theta_1) \cdots y(\bm\theta_{T})]^{\mathsf{T}}$ and $\bm f_T=[f(\bm\theta_1)  \cdots f(\bm\theta_{T})]^{\mathsf{T}}$. Specifically, 
\begin{equation}\label{eq:gain}
g_T=\mathrm{H}[\bm {y}_T]-\mathrm{H}[\bm {y}_T|\bm f_T],
\end{equation}
where $\mathrm{H}[\bm {y}_T]$ represents the information entropy of $\bm {y}_T$ and $\mathrm{H}[\bm {y}_T|\bm f_T]$ denotes the conditional information entropy of $\bm {y}_T$ given $\bm f_T$.

\section{Proof of Lemma~\ref{lemma:Lipschitz Continuity 1}}\label{appendix:Lemma1}
In this section, we present a complete proof of Lemma~\ref{lemma:Lipschitz Continuity 1} through a sequence of lemmas. We initially establish the following result regarding the partial derivative $\partial_jf(\bm\theta)$ of the noiseless QAOA objective function $f(\bm\theta): \mathcal{D}=[0,2\pi]^{2p}\mapsto\mathbb{R}$ for any $j\in[2p]$ and any $\bm\theta\in\mathcal{D}$.

\begin{lemma}\label{lemma:3}
Assuming that Assumption~\ref{assump:1-design} holds, let $f(\bm\theta): \mathcal{D}=[0,2\pi]^{2p}\mapsto\mathbb{R}$ be the noiseless QAOA objective function. Given a failure probability $\delta\in(0,1)$, the partial derivative $\partial_j f(\bm\theta)$ satisfies
\begin{equation}
 \forall j\in[2p],\forall\bm\theta\in\mathcal{D}, ~\left|\partial_jf(\bm\theta)\right|\leq \sqrt{\mathbb{V}_{\bm\theta}[\partial_a f(\bm\theta)]/\delta}
\end{equation}
with a success probability of at least $\geq1-\delta$, 
where $\mathbb{V}_{\bm\theta}[\partial_a f(\bm\theta)]$ is the variance of $\partial_af(\bm\theta)$ with index $a=\arg \max_{j\in[2p]}(\sup_{\bm\theta\in
\mathcal{D}}\left|\partial_jf(\bm\theta)\right|)$.
\end{lemma}
\begin{proof}
Fix $a\in[2p]$, by Chebyshev's Inequality~\cite{kvanli2006concise}, we have
\begin{equation}
\Pr \{\forall\bm\theta\in\mathcal{D},\forall s>0,~ \left|\partial_a f(\bm\theta)-\mathbb{E}_{\bm\theta}[\partial _af(\bm\theta)]\right|\leq s\}\geq 1-\mathbb{V}_{\bm\theta}[\partial_a f(\bm\theta)]/s^2,
\end{equation}
where $\mathbb{E}_{\bm\theta}[\partial _af(\bm\theta)]$ and $\mathbb{V}_{\bm\theta}[\partial_a f(\bm\theta)]$ are the expectation value and the variance of $\partial_a f(\bm\theta)$.
Assuming that Assumption~\ref{assump:1-design} holds, we demonstrate that $\mathbb{E}_{\bm\theta}[\partial_af(\bm\theta)]=0$. The detailed proof can be found in Ref~\cite{cerezo2021cost}. This implies 
\begin{equation}
\Pr \{\forall\bm\theta\in\mathcal{D},\forall s>0,~\left|\partial_a f(\bm\theta)\right|\leq s\}\geq 1-\mathbb{V}_{\bm\theta}[\partial_a f(\bm\theta)]/s^2.
\end{equation}
By choosing $a=\arg \max_{j\in[2p]}(\sup_{\bm\theta\in
\mathcal{D}}\left|\partial_jf(\bm\theta)\right|)$, we have
\begin{equation}
\Pr\left\{\forall s>0,~ \sup_{\bm\theta\in
\mathcal{D}}\left|\partial_af(\bm\theta)\right|\leq s\right\}\geq1-\mathbb{V}_{\bm\theta}[\partial_a f(\bm\theta)]/s^2.
\end{equation}
The use of the index $a$ and the notation $\rm sup(\cdot)$ immediately implies
\begin{equation}
\Pr\{\forall j\in[2p],\forall \bm\theta\in\mathcal{D},\forall s>0,~\left|\partial_jf(\bm\theta)\right|\leq s\}\geq1-\mathbb{V}_{\bm\theta}[\partial_a f(\bm\theta)]/s^2.
\end{equation}
Let the failure probability $\delta=\mathbb{V}_{\bm\theta}[\partial_a f(\bm\theta)]/s^2\in(0,1)$, we have
\begin{equation}
\Pr\left\{\forall j\in[2p],\forall \bm\theta\in\mathcal{D},~\left|\partial_jf(\bm\theta)\right|\leq \sqrt{\mathbb{V}_{\bm\theta}[\partial_a f(\bm\theta)]/\delta}\right\}\geq1-\delta.
\end{equation}
\end{proof}

\begin{lemma}\label{lemma:4}
Given a noiseless QAOA objective function $f(\bm\theta): \mathcal{D}=[0,2\pi]^{2p}\mapsto\mathbb{R}$, we have
\begin{equation}
\forall\bm\theta,\bm\theta^{\prime}\in\mathcal{D},~\left|f(\bm\theta)-f(\bm\theta^{\prime})\right|\leq\max_{j\in[2p]}\left(\sup_{\bm\theta\in\mathcal{D}}\left|\partial_jf(\bm\theta)\right|\right)\Vert\bm\theta-\bm\theta^{\prime}\Vert_1, 
\end{equation}
where $\partial_jf(\bm\theta)$ is the partial derivative of $f(\bm\theta)$ for $j\in[2p]$.
\end{lemma}
\begin{proof} 
Let $\bm\theta$ be represented as $[\theta _{1}, \cdots, \theta_{2p}]^{\mathsf{T}}$. For any $\bm\theta,\bm\theta^{\prime}\in\mathcal{D}$, we have 
\begin{equation}
  \begin{split}   
  f(\bm\theta)-f(\bm\theta^{\prime})=&f(\theta_1,\cdots,\theta_{2p})-f(\theta_1^{\prime},\theta_2,\cdots,\theta_{2p})+\cdots+\\
  &f(\theta_1^{\prime},\cdots,\theta_{j-1}^{\prime},\theta_j,\cdots,\theta_{2p})-f(\theta_1^{\prime},\cdots, \theta_j^{\prime},\theta_{j+1},\cdots,\theta_{2p})+\cdots+\\
  &f(\theta_1^{\prime},\cdots,\theta_{2p-1}^{\prime},\theta_{2p})-f(\theta_1^{\prime},\cdots,\theta_{2p}^{\prime}).  
  \end{split}
\end{equation}
By Triangle Inequality, for any $\bm\theta,\bm\theta^{\prime}\in\mathcal{D}$, we have 
\begin{equation}
\begin{split}
\left|f(\bm\theta)-f(\bm\theta^{\prime})\right|\leq&\left|f(\theta_1,\cdots,\theta_{2p})-f(\theta_1^{\prime},\theta_2,\cdots,\theta_{2p})\right|+\cdots+\\
&\left|f(\theta_1^{\prime},\cdots,\theta_{j-1}^{\prime},\theta_j,\cdots,\theta_{2p})-f(\theta_1^{\prime},\cdots, \theta_j^{\prime},\theta_{j+1},\cdots,\theta_{2p})\right|+\cdots+\\
&\left|f(\theta_1^{\prime},\cdots,\theta_{2p-1}^{\prime},\theta_{2p})-f(\theta_1^{\prime},\cdots,\theta_{2p}^{\prime})\right|.    
\end{split}
\end{equation}
For any $j\in[2p]$, the partial derivative with respect to the problem-oriented Hamiltonian $H_1$
\begin{equation}
\partial_jf(\bm\theta)=i\langle\varphi_0\lvert U_-^{\dagger}[H_1,U^{\dagger}_+H_1U_+]U_-\rvert\varphi_0\rangle
\end{equation}
and the partial derivative with respect to the mixing Hamiltonian $H_2$
\begin{equation}
\partial_jf(\bm\theta)=i\langle\varphi_0\lvert U_-^{\dagger}[H_2,U^{\dagger}_+H_1U_+]U_-\rvert\varphi_0\rangle
\end{equation}
exist and are continuous on $\mathcal{D}=[0,2\pi]^{2p}$, where $U_{-}$ is the left slice circuit and $U_{+}$ is the right slice circuit of the variational parameter $\theta_j$ in the noiseless QAOA circuit $U(\bm\theta)$, and \(|\varphi_0\rangle\) is the initial state. Fix $[\theta_1^{\prime},\cdots,\theta_{j-1}^{\prime},\theta_{j+1},\cdots,\theta_{2p}]^{\mathsf{T}}\in[0,2\pi]^{2p-1}$, $f(\bm\theta)$ can be seen as an uni-variable function in $\theta_j$. By Lagrange’s Mean Value Theorem~\cite{sohrab2003basic}, for any $\theta_j,\theta_j^{\prime}\in[0,2\pi]$ and for any $[\theta_1^{\prime},\cdots,\theta_{j-1}^{\prime},\theta_{j+1},\cdots,\theta_{2p}]^{\mathsf{T}}\in[0,2\pi]^{2p-1}$ denoted as $\hat{\bm\theta}\in\mathcal{\hat{D}}$, we have
\begin{equation}
\left|f(\theta_1^{\prime},\cdots,\theta_{j-1}^{\prime},\theta_j,\cdots,\theta_{2p})-f(\theta_1^{\prime},\cdots,\theta_{j}^{\prime},\theta_{j+1},\cdots,\theta_{2p})\right|\leq L_{j,\hat{\bm\theta}}\left|\theta_j-\theta_j^{\prime}\right|,
\end{equation}
where $L_{j,\hat{\bm\theta}}=\sup_{\theta_j\in[0,2\pi]}\left|\partial_jf(\bm\theta)\right|$. In light of this, for any $\theta_j,\theta_j^{\prime}\in[0,2\pi]$ and for any $\hat{\bm\theta}\in\hat{\mathcal{D}}$, we have
\begin{equation}
\left|f(\theta_1^{\prime},\cdots,\theta_{j-1}^{\prime},\theta_j,\cdots,\theta_{2p})-f(\theta_1^{\prime},\cdots,\theta_{j}^{\prime},\theta_{j+1},\cdots,\theta_{2p})\right|\leq L_j\left|\theta_j-\theta_j^{\prime}\right|,
\end{equation}
where $L_j=\sup_{\hat{\bm\theta}\in \hat{\mathcal{D}}}L_{j,\hat{\bm\theta}}$. Therefore, for any $\bm\theta,\bm\theta^{\prime}\in\mathcal{D}$, we have
\begin{align}
\left|f(\bm\theta)-f(\bm\theta^{\prime})\right|&\leq L_1\left|\theta_1-\theta_1^{\prime}\right|+\cdots+L_{2p}\left|\theta_{2p}-\theta_{2p}^{\prime}\right|\\
&\leq\left(\max_{j\in[2p]}L_j\right)\sum_{j=1}^{2p}\left|\theta_j-\theta_j^{\prime}\right|\\
&=\max_{j\in[2p]}L_j\Vert\bm\theta-\bm\theta^{\prime}\Vert_1\\
&=\max_{j\in[2p]}\left(\sup_{\bm\theta\in\mathcal{D}}\left|\partial_jf(\bm\theta)\right|\right)\Vert\bm\theta-\bm\theta^{\prime}\Vert_1.
\end{align} 
\end{proof}
Given Lemma~\ref{lemma:3} and Lemma~\ref{lemma:4}, we come to Lemma~\ref{lemma:Lipschitz Continuity 1} straightforwardly.
\begin{proof}[Proof of Lemma~\ref{lemma:Lipschitz Continuity 1}]
By Lemma~\ref{lemma:3}, we pick $\delta\in(0,1)$ and have
\begin{equation}
\Pr\left\{\max_{j\in[2p]}\left(\sup_{\bm\theta\in
\mathcal{D}}\left|\partial_jf(\bm\theta)\right|\right)\leq \sqrt{\mathbb{V}_{\bm\theta}[\partial_a f(\bm\theta)]/\delta}\right\}\geq1-\delta,
\end{equation}
where $\mathbb{V}_{\bm\theta}[\partial_a f(\bm\theta)]$ is the variance of the partial derivative $\partial_af(\bm\theta)$ with index $a=\arg \max_{j\in[2p]}(\sup_{\bm\theta\in
\mathcal{D}}\left|\partial_jf(\bm\theta)\right|)$.
Substituting this into Lemma~\ref{lemma:4}, the statement holds.
\end{proof}

\section{Proof of Theorem~\ref{theorem:Theorem 1}}\label{appendix:Theorem1}

\begin{theorem}[Formal]
Given a constant threshold $\epsilon$, a failure probability $\delta\in(0,1)$ and 
an $n$-qubit noiseless QAOA objective function $f(\bm\theta): \mathcal{D}=[0,2\pi]^{2p}\mapsto\mathbb{R}$~(Eq.~\ref{eq:Objective Function 1}) induced by the circuit $U(\bm\theta)$~(Eq.~\ref{eq:noiseless QAOA circuit}) that satisfies Assumption~\ref{assump:1-design}, run Algorithm~\ref{alg:BO for QAOA} for $T={\rm poly}(n^{1/\epsilon^2})$ steps, where the scaling parameter $\eta_t$ for the acquisition function $\mathrm{UCB}_t(\bm\theta)$~(Eq.\ref{eq:UCB}) used in each step $t$ is predefined as
\begin{equation}
\label{eq:eta 1}
\eta_t=2\log(2\pi^2t^2/3\delta)+4p\log(8\pi pt^2\sqrt{\mathbb{V}_{\bm\theta}[\partial_a f(\bm\theta)]/\delta}). 
\end{equation}
If the parameter dimension
\begin{equation}
p\leq\tilde{\mathcal{O}}\left(\sqrt{\log n}\right),
\end{equation}
then the optimization error $r_T$~(Eq.\ref{eq:ERROR 1}) satisfies $r_T\leq \epsilon$ with a success probability of at least $1-\delta$. Here, $\mathbb{V}_{\bm\theta}[\partial_a f(\bm\theta)]$ is the variance of the partial derivative $\partial_a f(\bm\theta)$ with index $a=\arg \max_{j\in[2p]}(\sup_{\bm\theta\in
\mathcal{D}}\left|\partial_j f(\bm\theta)\right|)$.
\end{theorem}

\subsection{Outline of the Proof Procedure}
Our objective is to determine the effective parameter dimension $p$ of the noiseless QAOA circuit $U(\bm\theta)$ such that the optimization error $r_T=f(\bm\theta ^*)-f(\bm\theta^+_T)$ after $T={\rm poly}(n)$ steps of executing BO can be upper bounded by a constant threshold $\epsilon$. Here, $\bm\theta^{*}$ represents the global maximum point and $\bm\theta_T^+$ denotes the approximation of the maximum point in the previous $T$ steps. We investigate this question through the lens of the Bayesian approach, which considers the corresponding noiseless QAOA objective function $f(\bm\theta)$ as a sample drawn from a Gaussian process with the Matern covariance function $k_{\mathrm{Matern}-\nu}(\bm\theta,\bm\theta^{\prime})$~(Eq.~\ref{eq:Matern}). We first establish that $r_T$ is upper bounded by $\frac{1}{T}\sum_{t=1}^T\left(f(\bm\theta ^*)-f(\bm\theta_t)\right)$, where $\bm\theta_t$ represents the next point selected in each step $t$. It is evident that the condition $\frac{1}{T}\sum_{t=1}^T\left(f(\bm\theta ^*)-f(\bm\theta_t)\right)\leq \epsilon$ is sufficient to deduce the result $r_T\leq\epsilon$. Hence, by ensuring that the upper bound on $\frac{1}{T}\sum_{t=1}^T\left(f(\bm\theta ^*)-f(\bm\theta_t)\right)$ is no greater than $\epsilon$, we can determine the effective $p$ that guarantees $r_T\leq \epsilon$. Subsequently, we utilize the continuity property of the noiseless QAOA objective function $f(\bm\theta)$~(Lemma~\ref{lemma:Lipschitz Continuity 1}) to establish an upper bound on $\frac{1}{T}\sum_{t=1}^T\left(f(\bm\theta ^*)-f(\bm\theta_t)\right)$. 

The complete proof of Theorem~\ref{theorem:Theorem 1} is supported by a series of lemmas~(Lemma~\ref{lemma:5}-Lemma~\ref{lemma:11}). We will introduce how these lemmas are employed in our proof. For convenience, we initially present explanations of several notions that commonly occur in the following sections. Specifically, $\mathbb{V}_{\bm\theta}[\partial_a f(\bm\theta)]$ denotes the variance of the partial derivative $\partial_af(\bm\theta)$ with index $a=\arg \max_{j\in[2p]}(\sup_{\bm\theta\in\mathcal{D}}\left|\partial_jf(\bm\theta)\right|)$. Additionally, $\mu_{t-1}(\bm\theta)$ represents the posterior mean function of $f(\bm\theta)$ and $\sigma_{t-1}(\bm\theta)$ denotes the posterior standard deviation of $f(\bm\theta)$ based on the accumulated observations $\mathcal{S}_{t-1}$ from the previous $t-1$ steps.

To facilitate the analysis in the continuous domain $\mathcal{D}=[0,2\pi]^{2p}$, we discretize $\mathcal{D}$ into a finite grid $\mathcal{D}_t$ in each step $t$, as it has been employed in Ref~\cite{srinivas2012information}. Specifically, the size of $\mathcal{D}_t$ is determined by the degree of discretization $\tau_t$, such that $\left|\mathcal{D}_t\right|=(\tau_t)^{2p}$. In the subsequent discussion, we use $[\bm\theta^*]_t$ to denote the closest point in $\mathcal{D}_t$ to $\bm\theta^*$. Next, we will evaluate upper bounds on $f(\bm\theta^*)-f([\bm\theta^*]_t)$ (the first term) and $f([\bm\theta^*]_t)$ (the second term) to obtain an upper bound on $f(\bm\theta^*)$. Regarding the first term, according to Lemma~\ref{lemma:5}, if $\tau_t=8\pi pt^2\sqrt{\mathbb{V}[\partial_a f(\bm\theta)]/\delta}$, then $f(\bm\theta^*)-f([\bm\theta^*]_t)$ can be upper bounded by $1/t^2$ with a success probability of at least $1-\delta/4$. Considering that $\bm\theta_t$ is selected by maximizing the acquisition function $\mathrm{UCB}_t(\bm\theta)$ over $\mathcal{D}$, according to Lemma~\ref{lemma:6}, $\mathrm{UCB}_t(\bm\theta_t)=\mu_{t-1}(\bm\theta_t)+\sqrt{\eta_t}\sigma_{t-1}(\bm\theta_t)$ can be used to upper bound $f([\bm\theta^*]_t)$ with a success probability of at least $1-\delta/4$. Here, a predefined scaling parameter $\eta_t=2\log\left(2\pi^2 t^2\left|\mathcal{D}_t\right|/3\delta\right)$ is used. Taking the two upper bounds mentioned above into account, Lemma~\ref{lemma:7} demonstrates that $$f(\bm\theta^*)=\left(f(\bm\theta^*)-f([\bm\theta^*]_t)\right)+f([\bm\theta^*]_t)\leq1/t^2+\mu_{t-1}(\bm\theta_t)+\sqrt{\eta_t}\sigma_{t-1}(\bm\theta_t)$$ with a success probability of at least $1-\delta/2$. Furthermore, we establish that $f(\bm\theta_t)$ is lower bounded by $\mu_{t-1}(\bm\theta_t)-\sqrt{\eta_t^{\prime}}\sigma_{t-1}(\bm\theta_t)$ with a success probability of at least $1-\delta/2$ using Lemma~\ref{lemma:8}, where
$\eta_t^{\prime}=2\log(\pi^2 t^2/3\delta)$. Since $\eta_t\geq\eta_t^{\prime}$, we can also use $\mu_{t-1}(\bm\theta_t)-\sqrt{\eta_t}\sigma_{t-1}(\bm\theta_t)$ as a lower bound for $f(\bm\theta_t)$. Afterward, Lemma~\ref{lemma:9} establishes that $$f(\bm\theta^*)-f(\bm\theta_t)\leq 1/t^2+2\sqrt{\eta_t}\sigma_{t-1}(\bm\theta_t)$$ with a success probability of at least $1-\delta$. Then, Lemma~\ref{lemma:10} establishes a connection between the sum of posterior variances $\sum_{t=1}^T \sigma_{t-1}^2(\bm\theta_t)$ and the information gain $g_T$~(Eq.~\ref{eq:gain}). As $f(\bm\theta)$ is considered as a sample drawn from a Gaussian process with $k_{\mathrm{Matern}-\nu}(\bm\theta,\bm\theta^{\prime})$, we can bound $\sum_{t=1}^T \sigma_{t-1}^2(\bm\theta_t)$ by the upper bound $\mathcal{O}(T^{\frac{p}{v+p}}\log^{\frac{v}{v+p}}(T))$ on the maximal $g_T$ for $k_{\mathrm{Matern}-\nu}(\bm\theta,\bm\theta^{\prime})$ in Ref~\cite{vakili2021information}. By applying Cauchy-Schwarz Inequality~\cite{bityutskov2001bunyakovskii} and considering the non-decreasing property of $\eta_t$ as $t$ increases, we can substitute the form of $\eta_T$ to obtain the result stated in Lemma~\ref{lemma:11}
$$r_T\leq \mathcal{O}\left(\sqrt{p\log \left(p T^2(\mathbb{V}_{\bm\theta}[\partial_a f(\bm\theta)])^{1/2}\right) (\log T/T)^{\frac{\nu}{\nu+p}}}\right)$$ with a success probability of at least $1-\delta$. Finally, we obtain the effective $p$ by solving for this upper bound is no greater than a constant threshold $\epsilon$ with $T={\rm poly}(n^{1/\epsilon^2})$.

\subsection{Proof Details}
In this section, we provide a comprehensive introduction to the corresponding lemmas. 

\begin{lemma}\label{lemma:5}
Assuming that Assumption~\ref{assump:1-design} holds, let $f(\bm\theta): \mathcal{D}=[0,2\pi]^{2p}\mapsto\mathbb{R}$ be the $n$-qubit noiseless QAOA objective function. Given a failure probability $\delta\in(0,1)$ and a finite grid
$\mathcal{D}_t$ of size $\left|\mathcal{D}_t\right|=(\tau_t)^{2p}$ with the degree of discretization $\tau_t=4\pi pt^2\sqrt{\mathbb{V}[\partial_a f(\bm\theta)]/\delta}$ in each step $t$, run
Algorithm~\ref{alg:BO for QAOA} for $T={\rm poly}(n)$ steps. The following relationship 
\begin{equation}
\forall t\in[T],\forall\bm\theta\in\mathcal{D},~\left|f(\bm\theta)-f([\bm\theta]_t)\right|\leq 1/t^2 
\end{equation}
holds with a success probability of at least $1-\delta$, where $[\bm\theta]_t$ represents the closest point in $\mathcal{D}_t$ to $\bm\theta$.
\end{lemma}
\begin{proof}
By choosing a finite grid $\mathcal{D}_t$ of size $(\tau_t)^{2p}$ in each step $t$, for any $\bm\theta\in\mathcal{D}$ we have $\left\Vert\bm\theta-[\bm\theta]_t\right\Vert_1\leq4\pi p /\tau_t$. Given Lemma~\ref{lemma:Lipschitz Continuity 1},  we have
\begin{equation}
\Pr\left\{\forall t\in[T],\forall\bm\theta\in\mathcal{D},~\left|f(\bm\theta)-f([\bm\theta]_t)\right|\leq 4\pi p \sqrt{\mathbb{V}[\partial_a f(\bm\theta)]/\delta}/\tau_t\right\}\geq1-\delta,
\end{equation}
where the failure probability $\delta\in(0,1)$. Since $\tau_t=4\pi pt^2\sqrt{\mathbb{V}[\partial_a f(\bm\theta)]/\delta}$, then 
\begin{equation}
\Pr\left\{\forall t\in[T],\forall\bm\theta\in\mathcal{D},~\left|f(\bm\theta)-f([\bm\theta]_t)\right|\leq 1/t^2\right\}\geq1-\delta.
\end{equation} 
Furthermore, we consider $\mathbb{V}[\partial_a f(\bm\theta)]$ to be ${1/\rm poly}(n)$, as shown in Ref~\cite{park2023hamiltonian}. Additionally, we assume that parameter dimension $p$ is at most ${\rm poly}(n)$. In order to guarantee the degree of discretization $\tau_t$ of at least 1, we enforce a constraint that the number of steps $T={\rm poly}(n)$. This constraint is consistent with the scenario we are exploring.
\end{proof}

\begin{lemma}\label{lemma:6}
Given a failure probability $\delta\in(0,1)$, an $n$-qubit noiseless QAOA objective function $f(\bm\theta): \mathcal{D}=[0,2\pi]^{2p}\mapsto\mathbb{R}$ and a finite grid $\mathcal{D}_t\subset\mathcal{D}$ of size $\left|\mathcal{D}_t\right|$ in each step $t$, run
Algorithm~\ref{alg:BO for QAOA} for $T={\rm poly}(n)$ steps, where a scaling parameter $\eta_t$ for the acquisition function $\mathrm{UCB}_t(\bm\theta)$ used in each step $t$ is predefined as $\eta_t=2\log(\pi^2 t^2\left|\mathcal{D}_t\right|/6\delta)$. The following relationship 
\begin{equation}
 \forall t\in[T],\forall\bm\theta\in\mathcal{D}_t,~f(\bm\theta)\in\mathcal{C}_t(\bm\theta)   
\end{equation}
holds with a success probability of at least $1-\delta$, where $\mathcal{C}_t(\bm\theta)$ represents a confidence interval
$[\mu_{t-1}(\bm\theta)-\sqrt{\eta_t}\sigma_{t-1}(\bm\theta),~ \mu_{t-1}(\bm\theta)+\sqrt{\eta_t}\sigma_{t-1}(\bm\theta)]$.
\end{lemma}
\begin{proof}
Fix $t\in[T]$ and $\bm\theta\in\mathcal{D}_t$. Conditioned on accumulated observations $\mathcal{S}_{t-1}$ from the previous $t-1$ steps, the posterior distribution $f(\bm\theta)\sim N(\mu_{t-1}(\bm\theta),\sigma_{t-1}^2(\bm\theta))$.
Now, if $b\sim N(0,1)$, then
\begin{align}
\Pr\{b>w\}&=\exp(-w^2/2)(2\pi)^{-1/2}\exp\left(-(b-w)^2/2-w(b-w)\right)\\
&\leq \exp(-w^2/2)\Pr\{b>0\}\\
&=\frac{1}{2}\exp(-w^2/2)    
\end{align}
for $w>0$, since $\exp(-w(b-w))\leq 1$ for $b\geq w$. Using $b=(f(\bm\theta)-\mu_{t-1}(\bm\theta))/\sigma_{t-1}(\bm\theta)$ and $w=\sqrt{\eta_t}$, we have
\begin{equation}
\Pr \{f(\bm\theta)\notin\mathcal{C}_t(\bm\theta)\}\leq \exp(-\eta_t/2).    
\end{equation}   
Applying the union bound for $\bm\theta\in\mathcal{D}_t$, we have
\begin{equation}
\Pr \{ \forall \bm\theta\in\mathcal{
D}_t,~f(\bm\theta)\in\mathcal{C}_t(\bm\theta)\}\geq 1-\left|\mathcal{D}_t\right|\exp(-\eta_t/2).
\end{equation}
Given that $\left|\mathcal{D}_t\right|\exp(-\eta_t/2)=\delta/q_t$, where $\sum_{t\geq1}(1/q_t)=1$, $q_t>0$, by applying the union bound for $t\in\mathbb{N}$, the statement holds. For example, we can use $q_t=\pi^2t^2/6$.
\end{proof}

\begin{lemma}\label{lemma:7}
Assuming that Assumption~\ref{assump:1-design} holds, let $f(\bm\theta): \mathcal{D}=[0,2\pi]^{2p}\mapsto\mathbb{R}$ be the $n$-qubit noiseless QAOA objective function. Given a failure probability $\delta\in(0,1)$, run Algorithm~\ref{alg:BO for QAOA} for $T={\rm poly}(n)$ steps, where a scaling parameter $\eta_t$ for the acquisition function $\mathrm{UCB}_t(\bm\theta)$ used in each step $t$ is predefined as $\eta_t=2\log(\pi^2t^2/3\delta)+4p\log(4\pi pt^2\sqrt{2\mathbb{V}_{\bm\theta}[\partial_a f(\bm\theta)]/\delta})$. The following relationship
\begin{equation}
\forall t\in[T],~f(\bm\theta^*)\leq\mu_{t-1}(\bm\theta_t)+\sqrt{\eta_t}\sigma_{t-1}(\bm\theta_t)+1/t^2
\end{equation}
holds with a success probability of at least $1-\delta$, where $\bm\theta^*$ denotes the global maximum point and $\bm\theta_t$ represents the next point selected in each step $t$. 
\end{lemma}
\begin{proof}
Using the failure probability $\delta/2$ in Lemma~\ref{lemma:5}, for the global maximum point $\bm\theta^*$, we have 
\begin{equation}\label{eq:B27}
\Pr\{\forall t\in[T],~f(\bm\theta^*)-f([\bm\theta^*]_t)\leq 1/t^2 \}\geq1-\delta/2,
\end{equation}
where $[\bm\theta^*]_t$ denotes the closest point in $\mathcal{D}_t$ to $\bm\theta^*$. 
Here, a finite grid $\mathcal{D}_t$ of size $\left|\mathcal{D}_t\right|=(\tau_t)^{2p}$ with $\tau_t=4\pi p t^2\sqrt{2\mathbb{V}[\partial_a f(\bm\theta)]/\delta}$.
Then, applying Lemma~\ref{lemma:6} with the failure probability $\delta/2$, for $[\bm\theta^*]_t$, we have
 \begin{equation}
\Pr\{\forall t\in [T],~f([\bm\theta^*]_t)\leq\mu_{t-1}([\bm\theta^*]_t)+\sqrt{\eta_t}\sigma_{t-1}([\bm\theta^*]_t)\}\geq1-\delta/2,
\end{equation}
where $\eta_t=2\log(\pi^2 t^2\left|\mathcal{D}_t\right|/3\delta)$. As the next point $\bm\theta_t$ is selected by maximizing $\mathrm{UCB}_t(\bm\theta)$ in each step $t$, we have $\mathrm{UCB}_{t}([\bm\theta^{*}]_t)\leq \mathrm{UCB}_{t}(\bm\theta_{t})$. Then, we have
 \begin{equation}\label{eq:B29}
\Pr\{\forall t\in [T],~f([\bm\theta^*]_t)\leq\mu_{t-1}(\bm\theta_t)+\sqrt{\eta_t}\sigma_{t-1}(\bm\theta_t)\}\geq1-\delta/2.
\end{equation}
Taking Eq.~\ref{eq:B27} and Eq.~\ref{eq:B29} together, the statement holds since $(1-\delta/2)^2>1-\delta$. 
\end{proof}

\begin{lemma}\label{lemma:8}
Given a failure probability $\delta\in(0,1)$ and an $n$-qubit noiseless QAOA objective function $f(\bm\theta): \mathcal{D}=[0,2\pi]^{2p}\mapsto\mathbb{R}$, run Algorithm~\ref{alg:BO for QAOA} for $T={\rm poly}(n)$ steps, where a scaling parameter $\eta_t^{\prime}$ for the acquisition function $\mathrm{UCB}_t(\bm\theta)$ used in each step $t$ is predefined as $\eta_t^{\prime}=2\log(\pi^2 t^2/6\delta)$. The following relationship
\begin{equation}
\forall t\in[T], ~f(\bm\theta_t)\in\mathcal{C}_t(\bm\theta_t)
\end{equation}
holds with a success probability of at least $1-\delta$, where $\bm\theta_t$ represents the next point selected in each step $t$ and $\mathcal{C}_t(\bm\theta_t)$ denotes the confidence interval
$[\mu_{t-1}(\bm\theta_t)-\sqrt{\eta_t^{\prime}}\sigma_{t-1}(\bm\theta_t),~ \mu_{t-1}(\bm\theta_t)+\sqrt{\eta_t^{\prime}}\sigma_{t-1}(\bm\theta_t)]$.
\end{lemma}
\begin{proof}
Fix $t\in[T]$. Conditioned on $\mathcal{S}_{t-1}$ from the previous $t-1$ steps, for the next point $\bm\theta_t$ selected in each step $t$, the posterior distribution $f(\bm\theta_t)\sim N(\mu_{t-1}(\bm\theta_t),\sigma_{t-1}^2(\bm\theta_t))$.
Now, if $b\sim N(0,1)$, then $\Pr\{b>w\}\leq\frac{1}{2}\exp(-w^2/2)$
for $w>0$. Using $b=(f(\bm\theta_t)-\mu_{t-1}(\bm\theta_t))/\sigma_{t-1}(\bm\theta_t)$ and $w=\sqrt{\eta_t^{\prime}}$, we have
\begin{equation}
\Pr \{f(\bm\theta_t)\notin\mathcal{C}_t(\bm\theta_t)\}\leq \exp(-\eta_t^{\prime}/2).    
\end{equation}
Given that $\exp(-\eta_t^{\prime}/2)=\delta/q_t$, where $\sum_{t\geq1}(1/q_t)=1$, $q_t>0$, by applying the union bound for $t\in\mathbb{N}$, the statement holds. For example, we can use $q_t=\pi^2t^2/6$.
\end{proof}

\begin{lemma}\label{lemma:9}
Assuming that Assumption~\ref{assump:1-design} holds, let $f(\bm\theta): \mathcal{D}=[0,2\pi]^{2p}\mapsto\mathbb{R}$ be the $n$-qubit noiseless QAOA objective function. Given a failure probability $\delta\in(0,1)$, run Algorithm~\ref{alg:BO for QAOA} for $T={\rm poly}(n)$ steps, where a scaling parameter $\eta_t$ for the acquisition function $\mathrm{UCB}_t(\bm\theta)$ used in each step $t$ is predefined as $\eta_t=2\log(2\pi^2t^2/3\delta)+4p\log(8\pi pt^2\sqrt{\mathbb{V}[\partial_a f(\bm\theta)]/\delta})$. The following relationship
\begin{equation}
\forall t\in[T],~f(\bm\theta^*)-f(\bm\theta_t)\leq2\sqrt{\eta_t}\sigma_{t-1}(\bm\theta_t)+1/t^2
\end{equation}
holds with a success probability of at least $1-\delta$, where $\bm\theta^*$ denotes the global maximum point and $\bm\theta_t$ represents the next point selected in each step $t$. 
\end{lemma}
\begin{proof}
Using the failure probability $\delta/2$ in Lemma~\ref{lemma:7}, for the global maximum point $\bm\theta^*$, we have
\begin{equation}\label{eq:B33}
\Pr\{
\forall t\in[T],~f(\bm\theta^*)\leq\mu_{t-1}(\bm\theta_t)+\sqrt{\eta_t}\sigma_{t-1}(\bm\theta_t)+1/t^2\}\geq1-\delta/2
\end{equation}
with $\eta_t=2\log(2\pi^2t^2/3\delta)+4p\log(8\pi pt^2\sqrt{\mathbb{V}[\partial_a f(\bm\theta)]/\delta})$
in each step $t$. Then, using the failure probability $\delta/2$ in Lemma~\ref{lemma:8}, for the next point $\bm\theta_t$ selected in each step $t$, we have
\begin{equation}\label{eq:B34}
\Pr\{\forall t\in [T],~f(\bm\theta_t)\geq\mu_{t-1}(\bm\theta_t)-\sqrt{\eta_t^{\prime}}\sigma_{t-1}(\bm\theta_t)\}\geq1-\delta/2
\end{equation}
with $\eta_t^{\prime}=2\log(\pi^2 t^2/3\delta)$ in each step $t$. As the aforementioned  $\eta_t$ is greater than $\eta_t^{\prime}$ used here, choosing $\eta_t$ here is also valid. Taking Eq.~\ref{eq:B33} and Eq.~\ref{eq:B34} together, the proof is completed.
\end{proof}

\begin{lemma}\label{lemma:10}
Given an $n$-qubit noiseless QAOA objective function $f(\bm\theta): \mathcal{D}=[0,2\pi]^{2p}\mapsto\mathbb{R}$, run Algorithm~\ref{alg:BO for QAOA} for $T={\rm poly}(n)$ steps. Let $\mathcal{S}_T=\{(\bm\theta_{1}, y(\bm\theta_{1})), \cdots, (\bm\theta_{T},y(\bm\theta_{T}))\}$ be the accumulated observations from the previous $T$ steps, where the estimation $y(\bm\theta_t)= f(\bm\theta_t) + \xi_t^{\rm noise}$ in each step $t$. Here, $\xi_t^{\rm {noise}}\sim N(0,1/4M)$ is independent and identically distributed Gaussian noise with $M$ representing the fixed number of measurements. The information gain $g_T$~(Eq.~\ref{eq:gain}) can be expressed as
\begin{equation}
g_T=\frac{1}{2}\sum_{t=1}^{T}\log(1+4M\sigma_{t-1}^{2}(\bm\theta_t)).    
\end{equation}
\end{lemma}
\begin{proof}
Let $\bm {y}_{t-1}=[y(\bm\theta_1) \cdots y(\bm\theta_{t-1})]^{\mathsf{T}}$ and $\bm f_{t-1}=[f(\bm\theta_1)  \cdots f(\bm\theta_{t-1})]^{\mathsf{T}}$ for $t\in[T+1]$. 
Plugging in the differential entropy of a multivariate Gaussian random variable~(Eq.~\ref{eq:entropy}), we have $\mathrm{H}[y(\bm\theta_t)|\bm {y}_{t-1}]=1/2\log(2\pi e(1/4M+\sigma^2_{t-1}(\bm\theta_t)))$ for $t\in[T]$ and $\mathrm{H}[\bm {y}_T|\bm  f_T]=\frac{T}{2}\log(\pi e/2M)$. Using the fact that $\mathrm{H}[\bm {y}_t]= \mathrm{H}[\bm {y}_{t-1}]+ \mathrm{H}[y(\bm\theta_t)|\bm {y}_{t-1}]$, we have
\begin{align}
\mathrm{H}[\bm {y}_T]&=\mathrm{H}[\bm {y}_0]+\mathrm{H}[y(\bm\theta_1)|\bm {y}_0]+ \mathrm{H}[y(\bm\theta_2)| \bm {y}_{1}]+\cdots+\mathrm{H}[y(\bm\theta_T)|\bm {y}_{T-1}]\\
&=\frac{1}{2}\sum_{t=1}^T\log(2\pi e(1/4M+\sigma^2_{t-1}(\bm\theta_t))).
\end{align}
Recalling the definition of $g_T$~(Eq.~\ref{eq:gain}), the statement holds.
\end{proof}

\begin{lemma}\label{lemma:11}
Assuming that Assumption~\ref{assump:1-design} holds, let $f(\bm\theta): \mathcal{D}=[0,2\pi]^{2p}\mapsto\mathbb{R}$ be the $n$-qubit noiseless QAOA objective function. Given a failure probability $\delta\in(0,1)$, run Algorithm~\ref{alg:BO for QAOA} with the Matern prior covariance function $k_{\rm Matern-\nu}(\bm\theta,\bm\theta^{\prime})$~(Eq.~\ref{eq:Matern}) for $T={\rm poly}(n)$ steps, where a scaling parameter $\eta_t$ for the acquisition function $\mathrm{UCB}_t(\bm\theta)$ used in each step $t$ is predefined as $\eta_t=2\log(2\pi^2t^2/3\delta)+4p\log(8\pi pt^2\sqrt{\mathbb{V}_{\bm\theta}[\partial_a f(\bm\theta)]/\delta})$. The optimization error $r_T$ satisfies
\begin{equation}
    r_T\leq\mathcal{O}\left(\sqrt{p\log \left(p T^2(\mathbb{V}_{\bm\theta}[\partial_a f(\bm\theta)])^{1/2}\right) (\log T/T)^{\frac{\nu}{\nu+p}}}\right)
\end{equation}
with a success probability of at least $1-\delta$.
\end{lemma}
\begin{proof}
Noted that $\eta_t$ in Lemma~\ref{lemma:9} is non-decreasing. Since $0\leq4M\sigma_{t-1}^2(\bm\theta_t)\leq4Mk_{\mathrm{Matern}-\nu}(\bm\theta_t,\bm\theta_t)\leq 4M$, denoted as $4M\sigma_{t-1}^2(\bm\theta_t)\in[0,4M]$, we have $4M\sigma_{t-1}^2(\bm\theta_t)\leq(4M/\log(1+4M))\log(1+4M\sigma_{t-1}^2(\bm\theta_t))$. Moreover, Lemma~\ref{lemma:10} links 
the sum of the posterior variances $\sum_{t=1}^T \sigma_{t-1}^2(\bm\theta_t)$ to the information gain $g_T$. By Cauchy-Schwarz Inequality~\cite{bityutskov2001bunyakovskii}, we have 
\begin{align}
\left(\sum_{t=1}^T2\sqrt{\eta_t}\sigma_{t-1 }(\bm\theta_t)\right)^2&\leq  \sum_{t=1}^T4\eta_t\sum_{t=1}^T\sigma_{t-1 }^2(\bm\theta_t)\\
&\leq \frac{T\eta_T}{M}\sum_{t=1}^T(4M\sigma_{t-1}^2(\bm\theta_t))\\
&\leq \frac{4T \eta_T}{\log(1+4 M)}\sum_{t=1}^T \log(1+4M\sigma_{t-1}^2(\bm\theta_t))\\
&= c_0 T \eta_Tg_T,   
\end{align}
where the parameter $c_0=8/\log(1+4M)$. 
The optimization error is given by
$r_T=f(\bm\theta ^*)-f(\bm\theta^+_T)$, where $\bm\theta^{*}$ represents the global maximum point and $\bm\theta_T^+=\arg \max_{\bm\theta\in\mathcal{A}_T} f(\bm\theta)$ denotes the approximation of the maximum point with the accumulated points $\mathcal{A}_{T}=\{\bm\theta_1,\cdots,\bm\theta_{T}\}$ from the previous $T$ steps. Now, we have
\begin{align}
r_T&\leq \frac{1}{T}\sum_{t=1}^T(f(\bm\theta^*)-f(\bm\theta_t))\\
&\leq\frac{1}{T}\left(\sum_{t=1}^T2\sqrt{\eta_t}\sigma_{t-1 }(\bm\theta_t)+\sum_{t=1}^T 1/t^2\right)\\
&\leq\frac{1}{T}\left(\sqrt{c_0 T\eta_T g_T}+\pi ^2/6\right).\label{eq:bound}
\end{align}
As $f(\bm\theta)$ is considered as a sample drawn from a Gaussian process with $k_{\mathrm{Matern}-\nu}(\bm\theta,\bm\theta^{\prime})$, we can use the upper bound $\mathcal{O}(T^{\frac{p}{v+p}}\log^{\frac{v}{v+p}}(T))$ on the maximal $g_T$ for $k_{\mathrm{Matern}-\nu}(\bm\theta,\bm\theta^{\prime})$ in Ref~\cite{vakili2021information}. By substituting $\eta_T$ and $\mathcal{O}(T^{\frac{p}{v+p}}\log^{\frac{v}{v+p}}(T))$ into Eq.~\ref{eq:bound}, the statement holds.  
\end{proof}

Now we are ready to complete the proof of Theorem~\ref{theorem:Theorem 1}.

\begin{proof}[Proof of Theorem~\ref{theorem:Theorem 1}] 
We consider $\mathbb{V}[\partial_a f(\bm\theta)]$ to be ${1/\rm poly}(n)$, as shown in Ref~\cite{park2023hamiltonian}. Additionally, we assume that the parameter dimension $p$ is at most ${\rm poly}(n)$. To ensure consistency with the scenario under investigation and to guarantee the degree of discretization $\tau_t$ in Lemma~\ref{lemma:5} of at least 1, we impose a constraint that the number of steps $T={\rm poly}(n)$. Hence, it is reasonable to treat $\log \left(p T^2(\mathbb{V}_{\bm\theta}[\partial_a f(\bm\theta)])^{1/2}\right)$ as a constant. Therefore, our task is to find the effective $p$ that satisfies the condition $(p(\log(T)/T)^{\frac{\nu}{\nu+p}})^{1/2}\leq\epsilon$, where $\epsilon$ is a constant threshold and $T={\rm poly}(n)$. Let 
\begin{align}\label{eq:effective p}
p\leq\frac{1}{2}\left(\epsilon^2-\nu+\sqrt{(\epsilon^2-\nu)^2+4\nu\epsilon^2\left(1+\log(T/\log T)\right)}\right),
\end{align}
then the above upper bound satisfies the inequality
\begin{align}
p^2-(\epsilon^2-\nu)p-\nu\epsilon^2\left(1+\log\left(T/\log T\right)\right)\leq 0.
\end{align}
Equivalently, the above inequality can be rewritten by
\begin{align}
\log\left(\log T/T\right)\leq \left(1+p/\nu\right)\left(1-p/\epsilon^2\right).
\end{align}
Considering the relationship $\log x\geq 1-1/x$ holds for $x>0$, then the above inequality implies
\begin{align}
\log\left(\log T/T\right)\leq \left(1+p/\nu\right)\log\left(\epsilon^2/p\right),
\end{align}
which directly leads to
\begin{align}
\log T/T\leq\left(\epsilon^2/p\right)^{1+p/\nu},
\end{align}
that is $(p(\log(T)/T)^{\frac{\nu}{\nu+p}})^{1/2}\leq\epsilon$.
Finally, let $T={\rm poly}(n^{1/\epsilon^2})$ and substitute it into Eq.~\ref{eq:effective p}. We obtain the effective parameter dimension $p$ for the noiseless QAOA circuit, which is $p\leq\tilde{\mathcal{O}}\left(\sqrt{\log n}\right)$.
\end{proof}

\section{Proof of Lemma~\ref{lemma:Lipschitz Continuity 2}}\label{appendix:Lemma2}
In this section, we provide the proof of Lemma~\ref{lemma:Lipschitz Continuity 2} which is similar to the proof of Lemma~\ref{lemma:Lipschitz Continuity 1}.

\begin{proof}[Proof of Lemma~\ref{lemma:Lipschitz Continuity 2}] 
Given an $n$-qubit noisy QAOA objective function with $q$-strength local Pauli channels $\tilde {f}_{q}(\bm\theta): \mathcal{D}=[0,2\pi]^{2p}\mapsto\mathbb{R}$, for any $j\in[2p]$, the partial derivatives $\partial_j\tilde {f}_{q}(\bm\theta)$ exist and are continuous, as shown in Ref~\cite{wang2021noise}. Using a similar proof sketch as in Lemma~\ref{lemma:4}, we have
\begin{equation}
\forall\bm\theta,\bm\theta^{\prime}\in\mathcal{D},~\left|\tilde {f}_{q}(\bm\theta)-\tilde {f}_{q}(\bm\theta^{\prime})\right|\leq \max_{j\in[2p]}\left(\sup_{\bm\theta\in\mathcal{D}}\left|\partial_j\tilde {f}_{q}(\bm\theta)\right|\right)\Vert\bm\theta-\bm\theta^{\prime}\Vert_1. 
\end{equation}
Considering the Maximum Cut problem on an unweighted $d$-regular graph with $n$ vertices, we can rely on Corollary~2 in Ref~\cite{wang2021noise} to obtain an upper bound on $\partial_j\tilde {f}_{q}(\bm\theta)$ for any $j\in[2p]$. Then, the following relationship 
\begin{equation}
\forall\bm\theta,\bm\theta^{\prime}\in\mathcal{D},~\left|\tilde f_q(\bm\theta)-\tilde f_q(\bm\theta^{\prime})\right|\leq L\Vert\bm\theta-\bm\theta^{\prime}\Vert_1  
\end{equation}
holds, where the Lipschitz continuity factor is given by 
\begin{equation}\label{eq:L}
 L=\sqrt{\ln 2/2} d^2 n^{\frac{5}{2}} \Vert H_1^{\rm {MaxCut}}\Vert_{\infty} q^{((d_1+1)p+1)}
\end{equation}
with the noise strength $q\in(0,1)$ and $d_1$ representing the circuit depth of the
implementation of the unitary corresponding to the problem-oriented Hamiltonian $H_1^{\rm MaxCut}$. Since $\Vert H_1^{\rm {MaxCut}}\Vert_{\infty}=\mathcal{O}(nd/2)$, $q\in(0,1)$ and  $d_1=\Omega(d)$, we have $L=\mathcal{O}(d^{3}n^{7/2}q^{(d+1)p})$. Thus, the proof of Lemma~\ref{lemma:Lipschitz Continuity 2} is concluded.
\end{proof}
 
\section{Proof of Theorem~\ref{theorem:Theorem 2}}\label{appendix:Theorem2}

\begin{theorem}[Formal]
Consider the Maximum Cut problem on an unweighted $d$-regular graph with $n$ vertices, where $d$ is a constant. 
Given a constant threshold $\epsilon$, a failure probability $\delta\in(0,1)$ and 
a noisy QAOA objective function with $q$-strength local Pauli channels $\tilde {f}_q(\bm\theta): \mathcal{D}=[0,2\pi]^{2p}\mapsto\mathbb{R}$~(Eq.~\ref{eq:Objective Function 2}) induced by the circuit $\mathcal{U}_{q}(\bm\theta)$~(Eq.~\ref{eq:noisy QAOA circuit}) that satisfies Assumption~\ref{assump:local Pauli noise channel}, run Algorithm~\ref{alg:BO for QAOA} for $T={\rm poly}(n^{1/\epsilon^2})$ steps, where the scaling parameter $\eta_t$ for the acquisition function $\mathrm{UCB}_t(\bm\theta)$~(Eq.~\ref{eq:UCB}) used in each step $t$ is predefined as
\begin{equation}
\eta_t=2\log(\pi^2t^2/(3\delta))+4p\log(4\pi pt^2d^{3}n^{7/2}q^{(d+1)p}).
\end{equation}
Under the condition where the noise strength $q$ spans $1/{\rm poly} (n)$ to $1/n^{1/\sqrt{\log n}}$, if the parameter dimension
 \begin{equation}
    p\leq\mathcal{O}\left(\log n/\log(1/q)\right),
\end{equation}     
then the optimization error $\tilde r_T$~(Eq.~\ref{eq:ERROR 2}) satisfies $\tilde r_T\leq\epsilon$ with a success probability of at least $1-\delta$.  
\end{theorem}

\subsection{Outline of the Proof Procedure}
Our objective is to determine the effective parameter dimension $p$ of the noisy QAOA circuit $\mathcal{U}_{q}(\bm\theta)$ such that the optimization error $\tilde r_T=\tilde f_{q}(\tilde{\bm\theta} ^*)-\tilde f_{q}(\tilde{\bm\theta}^+_T)$ after $T={\rm poly}(n)$ steps of executing BO can be upper bounded by a constant threshold $\epsilon$. Here, $\tilde{\bm\theta}^{*}$ represents the global maximum point and $\tilde{\bm\theta}_T^+$ denotes the approximation of the maximum point in the previous $T$ steps. We investigate this question through the lens of the Bayesian approach, which considers the corresponding noisy QAOA objective function $\tilde f_{q}(\bm\theta)$ as a sample drawn from a Gaussian process with the Matern covariance function $k_{\mathrm{Matern}-\nu}(\bm\theta,\bm\theta^{\prime})$. We first establish that $\tilde r_T$ is upper bounded by $\frac{1}{T}\sum_{t=1}^T(\tilde f_{q}(\tilde{\bm\theta} ^*)-\tilde f_{q}(\tilde{\bm\theta}_t))$, where $\tilde {\bm\theta}_t$ represents the next point selected in each step $t$. It is evident that the condition $\frac{1}{T}\sum_{t=1}^T(\tilde f_{q}(\tilde{\bm\theta} ^*)-\tilde f_{q}(\tilde{\bm\theta}_t))\leq \epsilon$ is sufficient to deduce the result $\tilde r_T\leq\epsilon$. Hence, by ensuring that the upper bound on $\frac{1}{T}\sum_{t=1}^T(\tilde f_{q}(\tilde{\bm\theta} ^*)-\tilde f_{q}(\tilde{\bm\theta}_t))$ is no greater than $\epsilon$, we can determine the effective $p$ that guarantees $\tilde r_T\leq \epsilon$. Subsequently, we utilize the continuity property of the noisy QAOA objective function $\tilde f_{q}(\bm\theta)$~(Lemma~\ref{lemma:Lipschitz Continuity 2}) to establish an upper bound on $\frac{1}{T}\sum_{t=1}^T(\tilde f_{q}(\tilde{\bm\theta} ^*)-\tilde f_{q}(\tilde{\bm\theta}_t))$. The complete proof of Theorem~\ref{theorem:Theorem 2} is similar to the proof of Theorem~\ref{theorem:Theorem 1} and is supported by a series of lemmas analogous to Lemma~\ref{lemma:5} to Lemma~\ref{lemma:11}. Instead of providing a detailed description of each lemma here, we will directly present lemma~\ref{lemma:5.1} similar to Lemma~\ref{lemma:11}. Additionally, we will emphasize the impact of the difference in continuity property between the noiseless and noisy QAOA objective functions on the result.

\subsection{Proof Details}
In this section, we provide a comprehensive introduction to the Lemma~\ref{lemma:5.1}.

\begin{lemma}\label{lemma:5.1}
Considering a Maximum Cut problem on an unweighted $d$-regular graph with $n$ vertices, where $d$ is a constant.
Assuming that Assumption~\ref{assump:local Pauli noise channel} holds, let $\tilde f_{q}(\bm\theta): \mathcal{D}=[0,2\pi]^{2p}\mapsto\mathbb{R}$ be the noisy QAOA objective function with $q$-strength local Pauli channels, where the noise strength $q\in(0,1)$. Given a failure probability $\delta\in (0,1)$, run Algorithm~\ref{alg:BO for QAOA} with the Matern prior covariance function $k_{\rm Matern-\nu}(\bm\theta,\bm\theta^{\prime})$~(Eq.~\ref{eq:Matern}) for $T={\rm poly}(n)$ steps, where a scaling parameter $\eta_t$ for the acquisition function $\mathrm{UCB}_t(\bm\theta)$ used in each step $t$ is predefined as $\eta_t=2\log(\pi^2t^2/(3\delta))+4p\log(4\pi pt^2d^{3}n^{7/2}q^{(d+1)p})$.
If the parameter dimension $p$ is given by
\begin{equation}
p\leq \mathcal{O}\left(\log n/\log(1/q)\right),    
\end{equation}
the optimization error $\tilde r_T$ satisfies
\begin{equation}
  \tilde r_T \leq \mathcal{O}\left(\sqrt{p\log (pT^2d^{3}n^{7/2}q^{(d+1)p})(\log T/T)^{\frac{\nu}{\nu+p}}}\right)  
\end{equation}
with a success probability of at least $1-\delta$. 
\end{lemma}
\begin{proof}
Using the continuity property of the noisy QAOA objective function $\tilde f_{q}(\bm\theta)$ as stated in Lemma~\ref{lemma:Lipschitz Continuity 2} and a series of lemmas similar to Lemma~\ref{lemma:5} to Lemma~\ref{lemma:11}, we can easily obtain the aforementioned result. It is essential to emphasize the constraint imposed on the parameter dimension $p$. To guarantee the degree of discretization $\tau_t$ of at least 1, as mentioned in Lemma~\ref{lemma:5}, we need to discuss the range of $p$ that satisfies $pT^2d^{3}n^{7/2}q^{(d+1)p}\geq 1$. Since the number of steps $T={\rm poly}(n)$ and $p$ is at most ${\rm poly}(n)$, we can establish the inequality 
\begin{equation}
 n^{c_2}\leq pT^2n^{7/2}\leq n^{c_1},   
\end{equation}
where $c_1$ and $c_2$ are two very close constants. Then, we have
\begin{equation}
   \frac{p}{n^{c_1}d^3} \leq \frac{1}{T^2n^{7/2}d^3}\leq \frac{p}{n^{c_2}d^3}.
\end{equation}
Since $q\in(0,1)$, the relationship
\begin{equation}
    q^{\frac{p(d+1)}{n^{c_2}d^3}}\leq q^{\frac{d+1}{T^2n^{7/2}d^3}}
\end{equation}
holds. As $y^y$ is monotonically decreasing in the interval $(0,1/e)$, we have
\begin{equation}
 \left(\frac{1}{p T^2n^{7/2}d^3 }\right)^{\frac{1}{p T^2n^{7/2}d^3 }}\leq \left(\frac{1}{n^{c_1}d^3}\right)^{\frac{1}{n^{c_1}d^3}}.
\end{equation}
Let 
\begin{equation}\label{eq:69}
 p\leq \frac{c_1\log n+3\log d}{(d+1)\log(1/q)n^{(c_1-c_2)}}  ,
\end{equation}
then the above inequality implies 
\begin{equation}
 \left(\frac{1}{n^{c_1}d^3}\right)^{\frac{1}{n^{c_1}d^3}}\leq q^{\frac{p(d+1)}{n^{c_2}d^3}} ,
\end{equation}
which directly leads to
\begin{equation}
   \left(\frac{1}{p T^2n^{7/2}d^3 }\right)^{\frac{1}{p T^2n^{7/2}d^3 }}\leq q^{\frac{d+1}{T^2n^{7/2}d^3}},
\end{equation}
that is $pT^2d^{3}n^{7/2}q^{(d+1)p}\geq 1$. Considering $d$ as a constant, Eq.~\ref{eq:69} implies $p\leq\mathcal{O}(\log n/\log(1/q))$. 
\end{proof}

\begin{proof}[Proof of Theorem~\ref{theorem:Theorem 2}]
Furthermore, when the noise strength $q\geq 1/{\rm poly}(n)$, it is reasonable to treat $\log (pT^2d^{3}n^{7/2}q^{(d+1)p})$ as a constant. Therefore, our objective is to determine the effective $p$ that satisfies the condition $(p(\log(T)/T)^{\frac{\nu}{\nu+p}})^{1/2}\leq\epsilon$ with a constant threshold $\epsilon$. The previous result shows that $p \leq \tilde{\mathcal{O}}(\sqrt{\log n})$ and $T={\rm poly}(n^{1/\epsilon^2})$. Therefore, we have 
\begin{equation}
  p\leq \min \{ \tilde{\mathcal{O}}(\sqrt{\log n}),\mathcal{O}(\log n/\log(1/q))\}.   
\end{equation}
Let $1/{\rm poly}(n)\leq q\leq 1/n^{1/\sqrt{\log n}}$, then this constraint implies
\begin{equation}
    \log n/\log(1/q)\leq \sqrt{\log n},
\end{equation}
that is $p\leq \mathcal{O}\left(\log n/\log(1/q)\right)$. Thus, the proof of Theorem~\ref{theorem:Theorem 2} is concluded.
\end{proof}

\end{document}